\documentclass[dvips,12pt]{article}
\usepackage{amsmath}
\usepackage{amsfonts}
\usepackage{amsthm}
\usepackage{amssymb}
\usepackage{amstext}
\usepackage{amsopn}
\usepackage{mathrsfs}
\usepackage{subfigure}
\usepackage{graphicx}
\usepackage[left=2.3cm,top=2cm,right=2.3cm,bottom=2cm, nohead,foot=1cm]{geometry}
\numberwithin{equation}{section}

\newtheorem{theorem}{Theorem}[section]
\newtheorem{lemma}[theorem]{Lemma}

\newtheorem{proposition}[theorem]{Proposition}

\newtheorem{corollary}[theorem]{Corollary}

\newcommand{\R}{{\mathbb R}}

\newcommand{\Hi}{{\mathcal{H}}}
\newcommand{\Bi}{{B}}

\newcommand{\N}{{\mathbb N}}
\newcommand{\Tr}{{\textup{Tr}}}

\newcommand{\scat}{{\mathbf{\mathit{S}}}}
\newcommand{\Sca}{{\mathbf{S}}}
\newcommand{\Asca}{{\mathbf{A}}}

\newcommand{\id}{{\textbf{\textup{id}}}}

\newcommand{\art}{{ \vec{a}_{\vec{k}, r,\lambda}(\vec{P})  }}

\newcommand{\arto}{{ \vec{a}_{\vec{k},r,\lambda}(\vec{P})  }}

\newcommand{\yrto}{{ \vec{a}_{\vec{k},r,\sigma,\lambda}(\vec{P})+(\sigma-I)\vec{k}  }}

\newcommand{\el}{{l}}

\newcommand{\kt}{{ \vec{d}_{\vec{k}, \lambda}(\vec{P})  }}

\newcommand{\kto}{{ \vec{d}_{\vec{k},\lambda}(\vec{P})  }}
\newcommand{\ktt}{{ \vec{d}_{\vec{k},\lambda}(\vec{P})   }}

\newcommand{\vt}{{ \vec{v}_{\vec{k},\sigma,\lambda}(\vec{P})  }}

\newcommand{\vot}{{ v_{k,\sigma_{2},\lambda}(P)     }}
\newcommand{\vto}{{ \vec{v}_{\vec{k},\sigma_{1},\lambda}(\vec{P})  }}
\newcommand{\vtt}{{ \vec{v}_{\vec{k},\sigma_{2},\lambda}(\vec{P})   }}
\newcommand{\vro}{{ v_{k,\sigma,r,\lambda}(P) }}
\newcommand{\vrt}{{ \vec{v}_{\vec{k},\sigma,r,\lambda}(\vec{P})  }}
\newcommand{\vroo}{{  v_{k,\sigma_{1},r,\lambda}(P)   }}

\newcommand{\vrto}{{ \vec{v}_{\vec{k},r,\sigma_{1},\lambda}(\vec{P})  }}
\newcommand{\vrtt}{{ \vec{v}_{\vec{k},r,\sigma_{2},\lambda}(\vec{P})   }}

\newcommand{\motj}{{m_{j,k,\sigma_{2},\lambda}(P)     }}

\newcommand{\mttj}{{ m_{j,\vec{k},\sigma_{2},\lambda}(\vec{P})   }}

\newcommand{\clt}{{ \mathbf{c}_{1,\sigma,r, \lambda}  }}

\newcommand{\clto}{{ \mathbf{c}_{\sigma_{1},r,\lambda}  }}

\newcommand{\crt}{{ \mathbf{c}_{2,\sigma,r, \lambda}  }}
\newcommand{\croo}{{  c_{2,\sigma_{1},r,\lambda}   }}

\newcommand{\crto}{{ \mathbf{c}_{2,\sigma_{1},r,\lambda}  }}

\newcommand{\cut}{{ \mathbf{c}_{3,\sigma,r, \lambda}  }}

\DeclareMathOperator*{\slim}{s-lim}

\title{The reduced effect of a single scattering with a low-mass particle via a point interaction}
\author{\textbf{Jeremy Clark}\\ jeremy.clark@fys.kuleuven.be \\ Katholieke Universiteit Leuven, Instituut voor Theoretische Fysica\\  Celestijnenlaan 200D, 3001 Heverlee, Belgium }

\begin{document}
\maketitle

\begin{abstract}
In this article, we study a second-order expansion for the effect induced on a large quantum particle which undergoes a single scattering with a low-mass particle via a repulsive point interaction.   We give an approximation with third-order error in $\lambda$  to the map $G\rightarrow \Tr_{2}[(I\otimes \rho)S_{\lambda}^{*}(G\otimes I)S_{\lambda}]$,    where $G\in \Bi(L^{2}(\R^{n}))$ is a heavy-particle observable, $\rho\in \Bi_{1}(\R^{n})$ is the density matrix corresponding to the state of the light particle, $\lambda=\frac{m}{M}$ is the mass ratio of the light particle to the heavy particle,  $S_{\lambda}\in \Bi(L^{2}(\R^{n})\otimes L^{2}(\R^{n}))$ is the scattering matrix between the two particles due to a repulsive point interaction, and the trace is over the light-particle Hilbert space.    The third-order error is bounded in operator norm for dimensions one and three using a weighted operator norm on $G$.     

 \end{abstract}

\section{Introduction}

 In theoretical physics, many  derivations of decoherence models begin with an analysis of the effect on a test particle of a scattering with a single particle from a background gas~\cite{Joos, ESL,Sipe}.  A regime that the theorists have studied and which has generated interest in experimental physics~\cite{Talbot} is when the test particle is much more massive than a single particle from the gas.  Mathematical progress towards justifying the scattering assumption made  in the physical literature in the regime where a test particle interacts with particles of comparatively low mass can be found in~\cite{AsymSysm,Carlone,Two}.     In this article, we study a scattering map  expressing  the effect induced on a test particle of mass $M$ by an interaction with a particle of mass $m=\lambda M$, $\lambda\ll 1$.   The force interaction  between the test particle and the gas particle is taken as a repulsive point potential.

We work towards bounding the error $\epsilon(G,\lambda)$ in operator norm for $G\in\Bi(L^{2}(\R^{n}))$, $n=1,3$  of a second order approximation:
\begin{align}\label{PoissonProc}
\Tr_{2}[ (I\otimes \rho) \Sca^{*}_{\lambda}(G\otimes I)\Sca_{\lambda} ]=G+\lambda M_{1}(G) +\lambda^{2}M_{2}(G)+\epsilon(G,\lambda),  
\end{align}
where $\rho\in  \Bi_{1}(L^{2}(\R^{n}))$ is a density matrix (i.e. $\rho\geq0$ and $\Tr[\rho]=1$), $G\in \Bi(L^{2}(\R^{n})) $, $\Sca_{\lambda}\in \Bi(L^{2}(\R^{n})\otimes L^{2}(\R^{n})) $ is the unitary scattering operator for a point interaction, and the partial trace is over the second component of the Hilbert space $L^{2}(\R^{n})\otimes L^{2}(\R^{n})$.   $M_{1}$ and $M_{2}$ are linear maps acting on a dense subspace of $\Bi(L^{2}(\R^{n}))$ ($M_{2}$ is unbounded).   Our main result is that there exists a $c>0$ such that for all $\rho$,  $G$, and $0\leq \lambda$ 
$$\|\epsilon(G,\lambda)\|\leq c\lambda^{3}\|\rho\|_{wtn}\|G\|_{wn},$$
where $\|\cdot\|_{wn}$ is a weighted operator norm of the form
\begin{multline*}
\|G \|_{wn} =\|G\|+\| |\vec{X}|G\|+\|G|\vec{X}|\|\\ +\sum_{0\leq i,j \leq d}(\|X_{i}P_{j}G\|+\|GP_{j}X_{i}\|)+\sum_{ e_{1}+e_{2}\leq 3}\| |\vec{P}|^{e_{1}}G |\vec{P}|^{e_{2}}\|,
\end{multline*}
and $\|\cdot \|_{wtn}$ is a weighted trace norm which will depend on the dimension.    In the above, $\vec{X}$ and $\vec{P}$ are the vector of position and momentum operators respectively:  $(X_{j}f)(x)=x_{j}f(x)$ and $(P_{j}f)(x)=i(\frac{\partial}{\partial x_{j}}f)(x)$.   Expressions of the type $A^{*}GB$ for unbounded operators $A$ and $B$ are identified with the kernel of the densely defined quadric form $F(\psi_{1};\psi_{2})=\langle A\psi_{1}|G B\psi_{2}\rangle$ in the case that $F$ is bounded. 

    The scattering operator is defined as  $\Sca_{\lambda}=(\Omega^{+})^{*}\Omega^{-}$,  where
\begin{align}
\Omega^{\pm}= \slim_{t \rightarrow \pm \infty }e^{itH_{tot} }  e^{-it H_{kin}}
\end{align}
are the M\"oller wave operators, and $H_{kin}$ is the kinetic Hamiltonian and is the standard self-adjoint extension of the sum of the Laplacians $-\frac{1}{2M}\Delta_{heavy}-\frac{1}{2m}\Delta_{light}$, while the total Hamiltonian $H_{tot}$ includes an additional repulsive point interaction between the particles.    The definition of $H_{tot}$ is a little tricky for $n>1$ since, in analogy to the Hamiltonian for a particle in a point potential~\cite{Solvable},  it can not be defined as a perturbation of $H_{kin}$ even in the sense of a quadratic form.   Rather, it is defined as a self-adjoint extension of $-\frac{1}{2M}\Delta_{heavy}-\frac{1}{2m}\Delta_{light}$ with a special boundary condition.   Going to center of mass coordinates, we can write
$$\frac{1}{2M}\Delta_{heavy}+\frac{1}{2m}\Delta_{light}= \frac{1}{2(m+M)}\Delta_{cm}+\frac{M+m}{2mM}\Delta_{dis},$$
so that the special boundary condition will be placed on the displacement coordinate corresponding to $\Delta_{dis}$ and  follows in analogy with that a single particle in a point potential as discussed in~\cite{Solvable}.   This also allows us to write down expressions for $\Sca_{\lambda}$.   Non-trivial point potentials in dimensions $>3$ do not exist and the main result of our analysis is restricted to dimensions one and three.

The first and second order expressions $M_{1}(G)$ and $M_{2}(G)$ respectively  have the form
\begin{align}\label{Lindblad}
M_{1}(G)= i[V_{1},G] \text{ and } M_{2}(G)= i[V_{2}+\frac{1}{2}\{\vec{A},\vec{P}\},G]+\varphi(G)-\frac{1}{2}\varphi(I)G-\frac{1}{2}G\varphi(I),
\end{align}
where $V_{1}$, $V_{2}$,  and $(\vec{A})_{j}$ for $j=1,\dots,n$ are bounded real-valued functions of the operator $\vec{X}$,  and $\varphi$ is a completely positive map admitting a Kraus decomposition: 
\begin{align}\label{Mult}
\varphi(G)=\sum_{j} \int_{\R^{3}}d\vec{k}\, m_{j,\vec{k}}^{*}\, G \,m_{j,\vec{k}},
\end{align}
with the $m_{j,\vec{k}}$'s being bounded multiplication operators in the $\vec{X}$-basis.  
Notice that terms in~(\ref{Lindblad}) are reminiscent of the form of a Lindblad generator~\cite{Lindblad}.   In~\cite{Clark} the results of this article are applied to the convergence of a quantum dynamical semigroup to a limiting form with generator including the terms~(\ref{Lindblad}).

 The explicit forms for $V_{1}$, $V_{2}$ , $\vec{A}$, and $\varphi $ are:
\begin{eqnarray}
V_{1}&=& \mathit{c}_{n}\, \mathit{s}_{n}^{-1} \int_{\R^{+}}dk \,  |k|^{-1}\int_{|\vec{v}_{1}|=|\vec{v}_{2}|=k } d\vec{v}_{1}\, d\vec{v}_{2}\,  \rho(\vec{v}_{1}, \vec{v}_{2} )\, e^{i\vec{X}(\vec{v}_{1}-\vec{v}_{2})}, \label{one}\\
V_{2}&=& \mathit{c}_{n}\,\mathit{s}_{n}^{-1}\int_{\R^{+}}dk\, |k|^{-1} \int_{|\vec{v}_{1}|=|\vec{v}_{2}|=k }d\vec{v}_{1}\, d\vec{v}_{2}\,   (\vec{v}_{1}+\vec{v}_{2})\nabla_{T}\rho(\vec{v}_{1},\vec{v}_{2})\, e^{i\vec{X}(\vec{v}_{1}-\vec{v}_{2})}, \label{two}\\
\vec{A}&=& \mathit{c}_{n}\,\mathit{s}_{n}^{-1}\int_{\R^{+}}dk\, |k|^{-1} \int_{|\vec{v}_{1}|=|\vec{v}_{2}|=k }d\vec{v}_{1}\, d\vec{v}_{2}\,  \nabla_{T}\rho(\vec{v}_{1},\vec{v}_{2} )\, e^{i\vec{X}(\vec{v}_{1}-\vec{v}_{2})},\label{three}\\
\varphi(G)&=& \mathit{c}_{n}^{2}\,\mathit{s}_{n}^{-2}\int_{\R^{n}}d\vec{k}\,|\vec{k}|^{-2}\int_{|\vec{v}_{1}|=|\vec{v}_{2}|=|\vec{k}| } d\vec{v}_{1}\, d\vec{v}_{2}\,   \rho(\vec{v}_{1}, \vec{v}_{2} )\, e^{i\vec{X}(-\vec{v}_{1}+\vec{k})}\, G \,e^{-i\vec{X}(-\vec{v}_{2}+\vec{k})},\label{four}
\end{eqnarray}
where $\mathit{s}_{n}$ is surface area of a unit ball in $\R^{n}$, $\mathit{c}_{n}$ is a constant arising form the scattering operator $\Sca_{\lambda}$, $\rho(\vec{k}_{1},\vec{k}_{2})$ is the integral kernel of $\rho$, and $\nabla_{T}$ is the gradient of weak derivatives in the diagonal direction which is formally $  (\nabla_{T}\rho(\vec{k}_{1},\vec{k}_{2} ))_{j}=\lim_{h\rightarrow 0} h^{-1}\big(\rho(\vec{k}_{1}+h\,e_{j},\vec{k}_{2}+h\,e_{j})-\rho(\vec{k}_{1},\vec{k}_{2})\big)  $.   The integral kernel $\rho(\vec{k}_{1},\vec{k}_{2})$ is well defined since $\rho$ is traceclass and hence Hilbert-Schmidt.    In dimension one, the integrals $\int_{|v_{1}|=|v_{2}|=k}$ are replaced by discrete sums.   In dimension two, $V_{2}$ has an additional term  due to the logarithm in~(\ref{ExpScat2}) which we did not write down in the expression for $V_{2}$ above.   The multiplication operators $m_{j,\vec{k}}$ are defined as $m_{j,\vec{k}}(\vec{X})=\mathit{c}_{n}s_{n}^{-1}\sqrt{\beta_{j}} \int_{|\vec{v}|=|\vec{k}|}d\vec{v}\, f_{j}(\vec{v})\,e^{-i\vec{X}(-\vec{v}+\vec{k})}$, where $\rho=\sum_{j}\beta_{j}|f_{j}\rangle \langle f_{j}|$ is the diagonalized form of $\rho$.     $V_{1}$, $V_{2}$, $\vec{A}$, and $\varphi$ are bounded under certain norm restrictions on $\rho$, since, for example, $\|V_{1}\|\leq \mathit{c}_{n}\| |\vec{P}|^{n-2}\rho \|_{1}$ and $\|\varphi\|=  \mathit{c}_{n}\| |\vec{P}|^{n-2}\rho |\vec{P}|^{n-2} \|_{1}$.

With the center of mass coordinate at the origin, the scattering operator $\Sca$ (neglecting the index $\lambda$ until it is explained) acts identically on the center-of-mass component of $L^{2}(\R^{2n})=L^{2}(\R^{n})\otimes L^{2}(\R^{n})$ as
 \begin{align}\label{ThisForm}
 \Sca=I+ I_{cm}\otimes ( \scat(k)\otimes |\phi\rangle \langle \phi| ),
 \end{align}
where the right copy of $L^{2}(\R^{n})$ corresponds to the  displacement variable and is decomposed in the momentum basis into a radial and an angular component as $L^{2}(\R^{+},r^{n-1}dr)\otimes L^{2}(\partial B_{1}(0) )$, $\scat(k)$ acts as a multiplication operator on the $L^{2}(\R^{+})$ component,   and $\phi=(s_{n})^{-\frac{1}{2}} 1_{\partial B_{1}(0)}$, is the normalized indicator function over the whole surface $\partial B_{1}(0)$.   We call $\scat(k)$ the scattering coefficient, and it has the form
\begin{align}\label{ListScat}
&\textbf{Dim}-1: \hspace{2.5cm}\textbf{Dim}-2: \hspace{4.4cm}\textbf{Dim}-3:  \hspace{2cm}\nonumber \\
&\scat_{\alpha} (k)=\frac{-i \alpha }{k+i\frac{1}{2} \alpha }\hspace{1.5cm}  \scat_{\mathit{l}}(k)=\frac{-i\pi}{ \mathit{l}^{-1}+\gamma+\ln(\frac{k}{2})+i \frac{\pi}{2} }\hspace{1.2cm} \scat_{\mathit{l}}(k)=\frac{-2ik}{ \mathit{l}^{-1} +ik},
\end{align}
 where $\alpha$ is a resonance parameter defined for the one-dimensional case, $\mathit{l}$ is the scattering length in the two- and three-dimensional cases and $\gamma \sim .57721$ is the Euler-Mascheroni constant.    In the one-dimensional case a scattering length $\mathit{l} $ is sometimes defined as the negative inverse of the resonance parameter $\alpha= \frac{\mu \mathit{c} }{\hbar^{2}}$, where $\mathit{c}$ is the coupling constant of the interaction and $\mu$ is the relative mass $\frac{mM}{m+M}=M\frac{\lambda}{1+\lambda}$.      However, this contrasts with the two- and three-dimensional cases where the scattering length is proportional to the strength of the interaction.     In the context of this article, where the point interaction is between a light and an heavy particle, we parameterize the resonance parameter as $\alpha=  \frac{\lambda}{1+\lambda} \alpha_{0}$ in the one-dimensional case and the scattering length as $\mathit{l}= \frac{\lambda}{1+\lambda} \mathit{l}_{0}$ in the two- and three-dimensional cases for some fixed $\alpha_{0}$ and $\mathit{l}$.   This corresponds to holding the strength of the interaction fixed.  Thus $\Sca_{\lambda}$ and $\scat_{\lambda}$ will be indexed by $\lambda$ for the remainder of the article.

There are two main obstacles in attempting to find a bound for the error $\epsilon(G,\lambda)$ from~(\ref{PoissonProc}).    The first obstacle is to find helpful expressions to facilitate making a Taylor expansion in $\lambda$ of $\Tr_{2}[ (I\otimes \rho) \Sca^{*}_{\lambda}(G\otimes I)\Sca_{\lambda} ]$.   Writing $\Asca_{\lambda}=\Sca_{\lambda}-I$, then 
$$\Tr_{2}[ (I\otimes \rho) \Sca^{*}_{\lambda}(G\otimes I)\Sca_{\lambda} ]=G+\Tr_{2}[(I\otimes \rho)\Asca_{\lambda}^{*}] G+G\Tr_{2}[(I\otimes \rho)\Asca_{\lambda}]+\Tr_{2}[(I\otimes \rho) \Asca_{\lambda}^{*}(G\otimes I)\Asca_{\lambda}],$$
and it turns out to be natural at all points of the analysis to approach the terms on the right individually. 
Propositions~\ref{GoodForm1} and~\ref{GoodForm2} are directed towards finding expressions for 
\begin{align}\label{WT}
\Tr_{2}[(I\otimes \rho)\Asca_{\lambda}^{*}] \text{ and  }\Tr_{2}[(I\otimes \rho) \Asca_{\lambda}^{*}(G\otimes I)\Asca_{\lambda}]
\end{align}
 respectively (since $\Tr_{2}[(I\otimes \rho)\Asca_{\lambda}]$ is merely the adjoint of $\Tr_{2}[(I\otimes \rho)\Asca_{\lambda}^{*}]$).    The expressions we find in Propositions~\ref{GoodForm1} and~\ref{GoodForm2} are of the form
  \begin{align}\label{GenDrib1}
\Tr_{2}[(I\otimes \rho)\Asca_{\lambda}^{*}]G= \int_{\R^{n}}d\vec{k}\int_{SO_{n}}d\sigma\,   U_{\vec{k},\sigma }^{*} f_{\vec{k},\sigma}^{*}\,G \text{, and }
 \end{align}
\begin{align} \label{GenDrib2}
\Tr_{2}[(I\otimes \rho) \Asca_{\lambda}^{*}(G\otimes I)\Asca_{\lambda}]=\int_{\R^{n}}d\vec{k}\int_{SO_{n}\times SO_{n}}d\sigma_{1}\, d\sigma_{2}\,U^{*}_{\vec{k},\sigma_{1},\lambda}\,h^{*}_{\vec{k},\sigma_{1},\lambda}\,G\, h_{\vec{k},\sigma_{2},\lambda}\, U_{\vec{k}, \sigma_{2},\lambda },
 \end{align}
 for some unitaries $U_{\vec{k},\sigma }^{*}$, $U_{\vec{k}, \sigma_{2},\lambda }$ and  some bounded operators  $g_{\vec{k},\sigma_{2},\lambda}$, $ h^{*}_{\vec{k},\sigma_{1},\lambda}$ which are functions of the vector of momentum operators $\vec{P}$.   In general, we will have the problem that
  $$\int_{\R^{n}}d\vec{k}\int_{SO_{n}}d\sigma\,  \| g_{\vec{k},\sigma,\lambda}\|=\infty, \text{ and }\int_{\R^{n}}d\vec{k}\int_{SO_{n}\times SO_{n}}d\sigma_{1}\,d\sigma_{2}\,\|h_{\vec{k},\sigma_{1},\lambda}\|\, \|h_{\vec{k},\sigma_{2},\lambda}\|=\infty,$$
so the integrals of operators only have strong convergence.    Propositions~\ref{IntBnd1} and~\ref{IntBnd2}  make sense of the integrals of operators such as~(\ref{GenDrib1}) and~(\ref{GenDrib2}) that arise and give operator norm bounds for the limits.   The basic pattern in the proof of Propositions~\ref{IntBnd1} and~\ref{IntBnd2} is an application of the simple inequalities in Propositions~\ref{StrongConv} and~\ref{StrongConv2} in addition to intertwining relations that we have between  the multiplication operators and the unitaries appearing in~(\ref{GenDrib1}) and~(\ref{GenDrib2}).

Bounding the third order error of the expansions in $\lambda$ of the strongly convergent integrals~(\ref{GenDrib1}) and~(\ref{GenDrib2})  brings up the second major obstacle.   We will need to bound certain strongly convergent integrals for all $\lambda$ in a neighborhood of zero.    For small $\lambda$ there will be unbounded expressions arising from the scattering coefficient $S_{\lambda}(k)$ that will have contrasting properties between the one- and three-dimensional cases.   For example, In the limit $\lambda\rightarrow 0$, $\frac{1}{\lambda}\scat_{\lambda}(k)$ becomes increasingly peaked in absolute value at $k\sim 0$ in the one-dimensional case.  For the three-dimensional case, $\frac{1}{\lambda}\scat_{\lambda}(k)$ becomes increasingly peaked at $k=\infty$.     A difficulty with the two-dimensional case  is the presence of the natural logarithm in the expression for $\scat_{\lambda}(k)$ and the fact that $\frac{1}{\lambda}\scat_{\frac{\lambda}{1+\lambda} \mathit{l}_{0}}(k)$ is not peaked at a fixed point as $\lambda$ varies.   The peak point  does tend towards $k\sim 0$ as $\lambda\rightarrow 0$, but  it is unknown how to attain the necessary inequalities in this case.

This article is organized as follows.    Section~\ref{ERESS} is concerned with proving Propositions~\ref{GoodForm1} and~\ref{GoodForm2} which give expressions for $\Tr[(I\otimes \rho)\Asca_{\lambda}^{*}]$ and  $\Tr[(I\otimes \rho) \Asca_{\lambda}^{*}(G\otimes I)\Asca_{\lambda}]$.   In Section~\ref{NonCom} we prove Propositions~\ref{IntBnd1} and~\ref{IntBnd2} which give the primary tools for bounding the integrals of operators which will arise in bounding the error term $\epsilon(G,\lambda)$ of our expansion~(\ref{PoissonProc}).   Section~\ref{Approx} contains the proof of Theorem~\ref{HardLem12} which is the main result of the article.    This involves expanding the expressions in Propositions~\ref{GoodForm1} and~\ref{GoodForm2} that we found in Section~\ref{ERESS} in $\lambda$ and bounding the error.   The difficult parts of the proof are characterized by using the Propositions~\ref{IntBnd1} and~\ref{IntBnd2} to translate unbounded expressions arising from the expansion of the scattering coefficient $\scat_{\lambda}$ into conditions on $G$ and $\rho$ through the weighted norms $\|G\|_{wn}$ and $\|\rho\|_{wtn}$ being finite.    Sections~\ref{ERESS} and~\ref{NonCom} apply to dimensions one through three (all dimensions where non-trivial point potentials exist), while Section~\ref{Approx} does not treat dimension two.

\section{Finding useful expressions for a single scattering }\label{ERESS}
In this section, we will find expressions for $\Tr_{2}[\rho \Asca^{*}_{\lambda}]$ and $\Tr_{2}[\rho \Asca_{\lambda}^{*}G\Asca_{\lambda}]$.   For notational convenience, we will begin identifying $I\otimes \rho$ with $\rho$ and $G\otimes I$ with $G$.  Finding formulas for $\Tr_{2}[\rho \Asca^{*}_{\lambda}]$, $\Tr_{2}[\rho \Asca_{\lambda}^{*}G\Asca_{\lambda}]$ begins with writing $\Asca_{\lambda}=\Sca_{\lambda}-I$ in a convenient way.    Let $f,g\in L^{2}(\R^{n}\times \R^{n})$, where the first and second component of $\R^{n}\times \R^{n}$ correspond to the displacement and the center of  mass coordinate, then
\begin{align}\label{Write}
\langle g|\Asca_{\lambda}  f\rangle=\int_{\R^{n}}d\vec{K}_{cm} \int_{0}^{\infty} dk\,  \frac{\scat_{\lambda}(k)}{\mathit{s}_{n} k^{n-1}} \big( \int_{\partial B_{k}(0)}d\hat{k}_{1}\, \bar{g}(\hat{k}_{1},\vec{K}_{cm}) \big)\big(\int_{\partial B_{k}(0)}d\hat{k}_{2}\, f(\hat{k}_{2},\vec{K}_{cm})\big).
\end{align}
The above formula gives a quadratic form representation of $\Asca_{\lambda}$ that involves integrating over a surface of  $3n-1$ degrees of freedom rather than $4n$, since it acts identically over the center-of-mass component of the Hilbert space and conserves energy for the complementary displacement coordinate.   The integral kernel for $\Asca_{\lambda}$ in center-of-mass momentum coordinates can be formally expressed as 
$$\Asca_{\lambda}(k_{dis,1},K_{cm,1}; k_{dis,2},K_{cm,2})= \frac{\scat_{\lambda}(|k_{dis,1}| )}{\mathit{s}_{n}|k_{dis,1}|^{n-1}}\delta(|k_{dis,1}|-|k_{dis,2}|)\delta(K_{cm,1}-K_{cm,2})$$
However, for instance, this does not work directly towards finding even a formal expression for 
\begin{align}\label{Read}
(\Tr_{2}[\rho \Asca^{*}_{\lambda}]G)(\vec{K}_{1},\vec{K}_2)=\int d\vec{k}_{1}\,d\vec{k}_{2}\, d\vec{K}\,\rho(k_{1},k_{2})\, \Asca^{*}_{\lambda}(\vec{k}_{2},\vec{K}_{1}; \vec{k}_{1},\vec{K})\,G(\vec{K},\vec{K}_{2}),
\end{align}
where we have written down a formal equation between integral kernel entries $\Tr_{2}[\rho \Asca^{*}_{\lambda}G]$, $G$, and $\Asca_{\lambda}^{*}$ using momentum coordinates corresponding to the heavy particle and the light particle.   In finding an expression for~(\ref{Read}), it would be  natural to have $\vec{K}_{2}$ as a parameterizing variable since the expression above is just multiplication of  $G$ from the left by $\Tr_{2}[\rho \Asca^{*}_{\lambda}]$. 
 
For $\lambda=\frac{m}{M}$, the center of mass coordinates  are $\vec{X}_{cm}=\frac{\lambda}{1+\lambda}\vec{x}+\frac{1}{1+\lambda}\vec{X}$ and $x_{d}=\vec{x}-\vec{X}$, where $\vec{x}$ and $\vec{X}$ are the position vectors of the particle with mass $m$ and $M$.   The corresponding momentum coordinates are $\vec{k}_{d}=\frac{1}{1+\lambda} \vec{k}-\frac{\lambda}{1+\lambda}\vec{K}$ and  $\vec{K}_{cm}=\vec{k}+\vec{K}$.    The proposition below gives two quadratic form representations of $\Asca_{\lambda}$ using different parameterisations  of the integration in~(\ref{Write}).  (\ref{GooodForm1}) is directed towards  finding an expression for $\Tr_{2}[\rho\Asca_{\lambda}^{*} ]$ and~(\ref{GooodForm2}) is for $\Tr_{2}[ \rho\Asca^{*}_{\lambda}G \Asca_{\lambda}]$.  The proof of the following proposition requires changes of integration.  

\begin{proposition}[Quadratic form representations of $\Asca_{\lambda} $] 
Let $f,g\in L^{2}(\R^{n}\times \R^{n})$, then
\begin{enumerate}
\item First Quadratic Form Representation
 \begin{multline}\label{GooodForm1}
\langle g|\Asca_{\lambda }    f\rangle=\int_{\R^{n}}d\vec{k}\,d\vec{K}\int_{ SO_{n}}d\sigma\, \scat_{\lambda }(|\vec{k}|)\\
 \bar{g}(\vec{k}+\lambda (\vec{K}+\sigma \vec{k}),\vec{K}+(\sigma-I)\vec{k})\,  f( \sigma \vec{k}+\lambda (\vec{K}+\sigma \vec{k}),\vec{K} ),
\end{multline}
\item Second Quadratic Form Representation
\begin{multline}\label{GooodForm2}
\langle g| \Asca_{\mathit{l}^{\prime} } f\rangle= \int d\vec{K}_{2}\, d\vec{k}_{1} \int_{ SO_{n} }d\sigma \,\det(I+\lambda \sigma)^{-1} \scat_{\lambda }(|\frac{I }{I+\lambda \sigma} ( \vec{k}_{1}-\lambda \vec{K}_{2}) | ) \\  \bar{g}(\vec{k}_{1},\vec{K}_{2}+\frac{(\sigma-I)}{1+\lambda \sigma} (\vec{k}_{1}-\lambda \vec{K}_{2})  )\, f( \vec{k}_{1}+\frac{\sigma-I}{1+\lambda \sigma} (\vec{k}_{1}-\lambda \vec{K}_{2}) , \vec{K}_{2} ),
\end{multline}
\end{enumerate}
where the total Haar measure on $SO_{n}$ is normalized to be $1$ (and for dimension one, the integral over $SO_{n}$ is replaced by a sum over $\{+,-\}$).

\end{proposition}

 The proofs of Propositions~\ref{GoodForm1} and~\ref{GoodForm2} work by using the spectral decomposition of $\rho$, special cases of $G$, etc. so that the quadratic form representations~(\ref{GooodForm1}) and~(\ref{GooodForm2})   of $\Asca_{\lambda}$ can be applied.  Defining $\tau_{\vec{k}}=e^{i\vec{k}\cdot\vec{X}}$, recall that $\tau_{\vec{k}}$ acts in the momentum basis as a shift: $(\tau_{\vec{k}} f)(\vec{p})=f(\vec{p}-\vec{k})$.

\begin{proposition}\label{GoodForm1}
Let $\rho$ have continuous integral operator elements in momentum representation.   $\Tr_{2}[\rho A^{*}_{\lambda}]$ has the integral form
\begin{align}\label{IntForm1}
\tilde{B}_{\lambda}^{*}=\int_{\R^{n}}d\vec{k} \int_{ SO_{n}}d \sigma\, \tau_{\vec{k}}\, \tau_{\sigma \vec{k}}^{*}\, p_{\vec{k},\sigma,\lambda}\, \bar{\scat}_{\lambda}(|\vec{k}|),
\end{align}
where $\tau_{\vec{a}}$ is a translation by $\vec{a}$ in the momentum $\vec{P}$ basis and  $p_{\vec{k},\sigma,\lambda}$ is a multiplication operator:
\begin{align*}
p_{\vec{k},\sigma,\lambda}=\rho( (1+\lambda)\vec{k}+\lambda \vec{P},(\sigma+\lambda)\vec{k}+\lambda \vec{P}).
\end{align*}

\end{proposition}

\begin{proof}

The following equality holds:
$$
\Tr_{2}[\rho \Asca^{*}_{\lambda}]=\Tr_{2}[\sum_{j}\beta_{j}|f_{j}\rangle \langle f_{j}| \Asca^{*}_{\lambda} ]=\sum_{j}\beta_{j}(\id \otimes \langle f_{j}|)\Asca^{*}_{\lambda}(\id \otimes |f_{j} \rangle),
$$
where the infinite sum on the right converges absolutely in the operator norm.  If we take a partial sum  $\rho_{m}=\sum_{j=1}^{m}\beta_{j}|f_{j}\rangle \langle f_{j}|$, then  using~(\ref{GooodForm1}),
\begin{multline*}
\sum_{j=1}^{m} \langle w| (\id \otimes \langle f_{j}|)\Asca^{*}_{\lambda} (\id \otimes |f_{j} \rangle) v\rangle= \sum_{j=1}^{m} \int_{\R^{n}\times \R^{n}} d\vec{K}_{1}d\vec{K}_{2}\int d\vec{k}\,  \bar{\scat}_{\lambda}(|\vec{k}|)  \int d\sigma \\ \bar{f}_{j}(\sigma \vec{k}+\lambda(\vec{K}_{1}+\sigma \vec{k}))\,\bar{w}(\vec{K}_{1})\,f_{j}(\vec{k}+\lambda (\vec{K}_{1}+\sigma\vec{k}))\,v(\vec{K}_{2}+(\sigma-I)\vec{k}).
\end{multline*}
 This has the form $\langle w| [\cdot] v\rangle$, where $[\cdot]$ is given by
$$
 \int_{\R^{n}}d\vec{k} \bar{\scat}_{\lambda}(\vec{k}) \int_{  SO_{n}}d\sigma\, \tau_{\sigma \vec{k}}^{*}\, \tau_{\vec{k}}\, \rho_{m}( (1+\lambda)\vec{k}+\lambda \vec{K} , (\sigma+\lambda)\vec{k}+\lambda \vec{K} ).
$$
  This converges in operator norm to the expression given by~(\ref{IntForm1}), since $\rho_{m}\rightarrow \rho$ in the trace norm and by  the bound given in Corollary~\ref{IntSpec1}.

\end{proof}

$\Tr_{2}[\rho\Asca_{\lambda}]$ has a similar integral representation by taking the adjoint.   Now we will delve into the form of $\Tr_{2}[\rho\Asca^{*}_{\lambda}G\Asca_{\lambda}]$.    In the following, the operator $D_{A}$ acts on $f\in L^{2}(\R^{n})$ as $(D_{A}f)(\vec{k})=|\det(A)|^{\frac{1}{2}}f(A\vec{k})$ for a element $A\in GL_{n}(\R)$.

\begin{proposition}\label{GoodForm2}
Let $\sum_{j}\beta_{j}|f_{j}\rangle\langle f_{j}|$ be the spectral decomposition of $\rho$.   $\Tr_{2}[\rho \Asca^{*}_{\lambda}G\Asca_{\lambda}]$ can be written in the form
\begin{multline}\label{FormHere}
\mathbf{\tilde{B}}_{\lambda}(G)= \sum_{j}\int_{\R^{n}}d\vec{k}\int_{SO_{n}\times SO_{n} }d\sigma_{1}d\sigma_{2}\, U^{*}_{\vec{k},\sigma_{1},\lambda}\, m_{j,\vec{k},\sigma_{1},\lambda}^{*}\, \bar{\scat}_{\lambda}\big(\big|\frac{\vec{k}-\lambda \vec{P}}{1+\lambda}\big|\big)\\   G\, \scat_{\lambda} \big(\big|\frac{\vec{k}-\lambda \vec{P}}{1+\lambda} \big|\big)\, m_{j,\vec{k},\sigma_{2},\lambda}\, U_{\vec{k},\sigma_{2},\lambda},
\end{multline}
where $
U_{\vec{k},\sigma_{2},\lambda}=\tau_{ k}^{*}D_{\frac{1+\lambda\sigma }{1+\lambda } } \tau_{\sigma\vec{k}},
$.   $\tau_{\sigma k}$, $\tau_{\vec{k}}$, and $D_{\frac{1+\lambda}{1+\lambda \sigma} }$ act on the momentum basis and $m_{j,\vec{k},\sigma,\lambda}$ is a function of the momentum operator $\vec{P}$ of the form
\begin{align*}
\sqrt{\beta_{j}} \det(1+\lambda \sigma )^{-\frac{1}{2}} f_{j}\big(\vec{k}+\frac{\sigma-I}{I+\lambda }(\vec{k}-\lambda\vec{P}) \big).
\end{align*}

\end{proposition}

\begin{proof}
Equation~(\ref{GooodForm2}) tells us how $\Asca_{\lambda}$ acts as a quadratic form.  In order to use~(\ref{GooodForm2}), we will  look at $\langle v| \Tr_{2}[\rho \Asca^{*}_{\lambda}G\Asca_{\lambda}] w\rangle$ in the special case where $G=G\otimes I=|y\rangle\langle y|\otimes I$ is a one-dimensional projection tensored with the identity over the light-particle Hilbert space.   Formally, this allows us to write
$$\langle v| \Tr_{2}[\rho \Asca^{*}_{\lambda}G\Asca_{\lambda}] w\rangle = \sum_{j}\sum_{l}\beta_{j} \langle v\otimes f_{j}| \Asca^{*}_{\lambda} |y\otimes \phi_{l}\rangle \langle y\otimes \phi_{l} |\Asca_{\lambda} | w\otimes f_{j} \rangle,   $$
where $(\phi_{m})$ is some orthonormal basis over the light-particle Hilbert space allowing a representation of the identity operator as a sum of one-dimensional projections , and the spectral decomposition of $\rho$ has been used.    Once~(\ref{GooodForm2}) has been applied, we build  up to an expression~(\ref{FormHere}), taking care with respect to the limits involved.    By Corollary~\ref{IntSpec2}, the expression~(\ref{FormHere}) defines a bounded completely positive map (c.p.m.).   Since $\Tr_{2}[\rho \Asca^{*}_{\lambda}G \Asca_{\lambda}]$ defines a c.p.m.  and agrees with~(\ref{FormHere}) for one-dimensional orthogonal projections, it follows that the two expressions are equal on $\Bi(L^{2}(\R^{d}))$.  This follows because c.p.m.'s are strongly continuous and the span of one-dimensional orthogonal projections is strongly dense.

The following holds, where the right-hand side converges in the operator norm:
  $$
\Tr_{2}[\rho \Asca^{*}_{\lambda}G \Asca_{\lambda}]=\sum_{j}\beta_{j}(\id \otimes \langle f_{j}|)\Asca^{*}_{\lambda}G \Asca_{\lambda} (\id \otimes |f_{j} \rangle).
$$
For $G=|y\rangle\langle y|$, $(\id \otimes \langle f_{j}|)\Asca^{*}_{\lambda}G\Asca_{\lambda} (\id \otimes |f_{j} \rangle)= \varphi_{y,j}(I)$, where $\varphi_{y,j}$ is the completely positive map such that for $H\in \Bi(\Hi)$
$$
\varphi_{y,j}(H)=(\id \otimes \langle f_{j}|)\Asca^{*}_{\lambda}( |y\rangle\langle y|\otimes H ) \Asca_{\lambda} (\id \otimes |f_{j}\rangle).
$$
Since $\varphi_{y,j}$ is completely positive, $\varphi_{y,j}(\sum_{l=1}^{m}|\phi_{l}\rangle\langle \phi_{l} |)$ converges strongly to $\varphi_{y,j}(I)$.    $\varphi_{y,j}(I)$ is determined by its expectations $\langle v| \varphi_{y}(I)v\rangle$, and moreover
\begin{multline}\label{Relate}
\langle v|\varphi_{y,j}(I)  v\rangle = \lim_{N\rightarrow \infty} \langle v |\varphi_{y,j}( \sum_{m=1}^{N}   |\phi_{m}\rangle \langle \phi_{m}|) v \rangle  \\ = \lim_{N\rightarrow \infty} \sum_{m=1}^{N}\langle \phi_{m}|   \upsilon_{v,j,y}\rangle \langle \upsilon_{v,j,y}|  \phi_{m}\rangle=  \|\upsilon_{v,j,y}\|^{2}=  \int_{\R^{n}}d\vec{k}\, \bar{\upsilon}_{v,j,y}(\vec{k})\, \upsilon_{v,j,y}(\vec{k}),
\end{multline}
where $ \upsilon_{v,j,y}$ is defined as the vector $\upsilon_{v,j,y}= (\langle y|\otimes \id )\Asca_{\lambda} ( |v\rangle \otimes |f_{j}\rangle) $.  Using~(\ref{GooodForm2}), $\langle \phi_{m}|\upsilon_{v,j,y}\rangle$ can be expressed as
\begin{multline}\label{FormHere2}
\langle \phi_{m}|\upsilon_{v,j,y}\rangle= \int d\vec{K}   \int_{ SO_{n} }d\sigma\, \scat_{\lambda}\big(\big|\frac{I }{I+\lambda\sigma }(\vec{k}-\lambda \vec{K}) \big| \big)\,\det(I+\lambda \sigma)^{-1}\\ \bar{\phi}_{m}(\vec{k})\, \bar{y}\big(\vec{K}+\frac{\sigma-I}{I+\lambda \sigma}(\vec{k}-\lambda \vec{K}) \big)\, f_{j}\big(\vec{k}+\frac{\sigma-I}{I+\lambda\sigma}(\vec{k}-\lambda \vec{K} \big)\,  v(\vec{K}).
\end{multline}
By~(\ref{Relate}), we can evaluate    $\Tr_{2}[(|f_{j}\rangle\langle f_{j}|)\Asca^{*}_{\lambda}(|y\rangle\langle y|\otimes I) \Asca_{\lambda}]= \langle v| \varphi_{y}(I) v\rangle$ through expression  $ \int_{\R^{n}}d\vec{k} \bar{\upsilon}_{v,j,y}(\vec{k}) \upsilon_{v,j,y}(\vec{k})$.   Through~(\ref{FormHere2}) we have an a.e. defined expression for the values $\upsilon_{v,j,y}(\vec{k})$.     Now, writing down $ \int_{\R^{n}}d\vec{k} \bar{\upsilon}_{v,j,y}(\vec{k}) \upsilon_{v,j,y}(\vec{k})$ using the expression for $\upsilon_{v,j,y}(\vec{k})$, the result can be viewed as an integral of operators acting from the left and the right on $|y\rangle\langle y|$, followed by an evaluation $\langle v| (\cdot) v\rangle$.   Using the intertwining relation:
\begin{align*}
m(\vec{P})\,\tau_{\vec{k}}^{*}\,D_{\frac{1+\lambda \sigma}{1+\lambda}}\,\tau_{\sigma \vec{k}}=\tau_{\vec{k}}^{*}\,D_{\frac{1+\lambda \sigma}{1+\lambda}}\,\tau_{\sigma \vec{k}}\, m(\vec{P}-\frac{\sigma-I}{1+\lambda \sigma}(\vec{k}-\lambda \vec{P}) ),
\end{align*}
for a function $m(\vec{P})$ of the momentum operators $\vec{P}$ and the fact that $\frac{\sigma+\lambda }{I+\lambda \sigma}=\sigma \frac{I+\lambda \sigma^{-1}}{I+\lambda \sigma}$ is an isometry for $0\leq \lambda < 1$, the expression can be written:
\begin{multline*}
\langle v| \varphi_{y,j}(I) v\rangle= \langle v| \int_{\R^{n}}d\vec{k}  \int_{SO_{n}\times SO_{n} }d\sigma_{1}\,d\sigma_{2}\, [U_{\vec{k},\sigma_{1},\lambda}^{*}\,m_{j,\vec{k},\sigma_{1},\lambda }^{*}\,   \bar{\scat}_{\lambda}\big(\big| \frac{\vec{k}-\lambda \vec{P} }{1+\lambda}\big|\big) \\  (|y\rangle \langle y|)\,\scat_{\lambda}\big(\big| \frac{\vec{k}-\lambda \vec{P}}{1+\lambda}\big|\big)\,  m_{j,\vec{k},\sigma_{2},\lambda}\, U_{\vec{k},\sigma_{2},\lambda}] | v\rangle.
\end{multline*}
So $\varphi_{y,j}(I)=\Tr_{2}[ (|f_{j}\rangle \langle f_{j}|) \Asca^{*}_{\lambda}(|v\rangle\langle v|)\Asca_{\lambda} ]$ agrees with the expression~(\ref{FormHere}) for a fixed $j$ and for $G=|v\rangle \langle v|$ for all $v$, and hence by our observation at the beginning of the proof, $\Tr_{2}[(|f_{j}\rangle\langle f_{j}|) \Asca^{*}_{\lambda}G\Asca_{\lambda}]$ is equal to the expression~(\ref{FormHere}) for a single fixed $j$ and all $G\in \Bi(L^{2}(\R^{n}))$.    However, if we take the limit $m\rightarrow \infty$  for $\rho_{m}=\sum_{j=1}^{m}\beta_{j}|f_{j}\rangle \langle f_{j}|$, then the expression~(\ref{FormHere}) converges in the operator norm and $\Tr_{2}[\rho_{m} \Asca^{*}_{\lambda}G\Asca_{\lambda} ]$ converges to $\Tr_{2}[\rho \Asca^{*}_{\lambda}G\Asca_{\lambda}]$.   Hence we have equality for all trace class $\rho$.

\end{proof}

Through the formula $\Tr_{2}[\rho \Sca^{*}_{\lambda}G\Sca_{\lambda}]=G+\tilde{B}^{*}G+G\tilde{B}+\mathbf{\tilde{B}}(G)$, it is clear that $\tilde{B}^{*}+\tilde{B}=-\mathbf{\tilde{B}}(I)$ by plugging in $G=I$.   However, it is not at all obvious that this equality takes place through the expressions~(\ref{IntForm1}) and~(\ref{FormHere}) for $\tilde{B}^{*}$ and $\mathbf{\tilde{B}}(I)$, respectively, since the operators $U_{\vec{k},\sigma,\lambda}$ appear only in form for $\mathbf{\tilde{B}}(I)$.

  It is convenient to notice the intertwining relation $h(\vec{k}-\lambda\vec{P})\,U_{\vec{k},\sigma,\lambda}= U_{\vec{k},\sigma,\lambda} \,h(\frac{1+\lambda}{I+\lambda\sigma}\vec{k}-\lambda\vec{P}))$.  Let $g\in L^{2}(\R^{n})$, then
 $\hat{g}=\mathbf{\tilde{B}}(I)g$ can be written:
\begin{multline}\label{GoodBad}
\hat{g}(\vec{p})=\sum_{j}\int_{\R^{n}}d\vec{k}\int_{SO_{n}\times SO_{n} }d\sigma_{1}\,d\sigma_{2}\, \big(U^{*}_{\vec{k},\sigma_{1},\lambda}\,  m_{j,\vec{k},\sigma_{1},\lambda}^{*}\,U_{\vec{k},\sigma_{1},\lambda}\big) (\vec{p})\\ |\scat_{\lambda} |^{2}\,\big(\big|\frac{I}{I+\lambda \sigma_{1} } (\vec{k}-\lambda \vec{p})\big|\big)\,\big(U^{*}_{\vec{k},\sigma_{1},\lambda} \, m_{j,\vec{k},\sigma_{2},\lambda}\,U_{\vec{k},\sigma_{1},\lambda}\big)(\vec{p})\,(U^{*}_{\vec{k},\sigma_{1},\lambda}U_{\vec{k},\sigma_{1},\lambda}g)(\vec{p}),
\end{multline}
where we have intertwined $U_{\vec{k},\sigma_{1},\lambda}^{*}$ from the left to the right, and
\begin{align*}
\big(U^{*}_{\vec{k},\sigma_{1},\lambda}\,  m_{j,\vec{k},\sigma_{1},\lambda}^{*}\, U_{\vec{k},\sigma_{1},\lambda}\big)(\vec{p})=   \sqrt{\beta_{j}}\,\det(I+\lambda \sigma_{1})^{-\frac{1}{2}}\,      \bar{f}_{j}\big(\vec{k}+\frac{\sigma_{1}-I}{I+\lambda \sigma_{1} }(\vec{k}-\lambda\vec{p}) \big),
\end{align*}
\begin{align*}
\big(U^{*}_{\vec{k},\sigma_{1},\lambda}\,  m_{j,\vec{k},\sigma_{2},\lambda}\,U_{\vec{k},\sigma_{1},\lambda}\big)(\vec{p})= \sqrt{\beta_{j}}\, \det( I+\lambda \sigma_{2} )^{-\frac{1}{2}}\, f_{j}\big(\vec{k}+(\sigma_{2}-I)(I+\lambda \sigma_{1})^{-1} (\vec{k}-\lambda\vec{p}) \big),
\end{align*}
$$\big(U^{*}_{\vec{k},\sigma_{1},\lambda}\,U_{\vec{k},\sigma_{1},\lambda}g\big)(\vec{p})=\det\big(\frac{1+\lambda }{I+\lambda\sigma_{1}  }\big)^{\frac{1}{2}}\,  \det\big(\frac{I+\lambda \sigma_{2} }{1+\lambda  }\big)^{\frac{1}{2}}\,  g(\vec{p}+(\sigma_{1}-\sigma_{2})(I+\lambda \sigma_{1})^{-1}(\vec{k}-\lambda\vec{p}) ).$$
 Making the change of variables $ \frac{\sigma_{1} }{I+\lambda \sigma_{1}}(\vec{k}-\lambda \vec{p})\rightarrow \vec{k}$, the resulting expression has only angular dependance of $\sigma_{2}\sigma_{1}^{-1}=\sigma$, and integrating out the other angular degrees of freedom yields $\mathbf{\tilde{B}}$.

\section{Bounding integrals of non-commuting operators}\label{NonCom}
Now we move on to proving Propositions~\ref{IntBnd1} and~\ref{IntBnd2} below which are proved in much greater generality than needed for this section, but they will serve as the principle tools in Section~\ref{Approx}.  To state these propositions we will need to generalize the concept of a multiplication operator.  Let $\Hi_{1}$, $\Hi_{2}$ be Hilbert spaces.  Given a bounded function $M:\R^{n}\rightarrow \mathcal{B}(\Hi_{1},\Hi_{2})$  we can construct an element   $\mathbf{M}\in \mathcal{B}(L^{2}(\R^{n})\otimes \Hi_{1}, L^{2}(\R^{n})\otimes \Hi_{2})$ using the equivalence $L^{2}(\R^{n})\otimes \Hi_{1} \cong L^{2}(\R^{n},\Hi_{1})$, where for $\mathbf{f}\in L^{2}(\R^{n})\otimes \Hi_{1} $
$$
\mathbf{M}(\mathbf{f})(\vec{x})=M(\vec{x})\mathbf{f}(\vec{x}).
$$
We will call these multiplication operators.

\begin{proposition}\label{IntBnd1}
Define $B:L^{2}(\R^{n})\otimes \Hi_{1} \rightarrow L^{2}(\R^{n})\otimes \Hi_{2}$, s.t.
\begin{align}\label{IntBnd11}
B=\int_{\R^{n}}d\vec{k} \int_{SO_{n}}d\sigma\, \tau_{\vec{k}}^{*}\,\tau_{\mathbf{a}_{\sigma} \vec{k} }\,q_{\vec{k} ,\sigma} ,
\end{align}
where $q_{\vec{k} ,\sigma} $ is a multiplication operator in the $\vec{P}$ basis of the form:
$$
q_{\vec{k} ,\sigma} =n_{\vec{k} ,\sigma}(\vec{P}) \eta(\mathbf{x}_{1,\sigma}\vec{k}+\mathbf{y}_{\sigma}\vec{P} , \mathbf{x}_{2,\sigma}\vec{k}+\mathbf{y}_{\sigma}\vec{P}  ),
$$
where $\eta(\vec{k}_{1},\vec{k}_{2})$ is continuous and defines a trace class integral operator on $L^{2}(\R^{n})$,\\ $\mathbf{a}_{\sigma},\mathbf{x}_{1,\sigma},  \mathbf{x}_{2,\sigma}, \mathbf{y}_{\sigma}\in M_{n}(\R)$, and
 $n_{\vec{k} ,\sigma}\in \mathcal{B}(L^{2}(\R^{n})\otimes \Hi_{1},L^{2}(\R^{n})\otimes \Hi_{2} )$ is a multiplication operator.  Let
 $$|\det(\mathbf{x}_{1,\sigma}+\mathbf{y}_{\sigma}(\mathbf{a}_{\sigma}-I)   )|, |\det(\mathbf{x}_{2,\sigma}+\mathbf{y}_{\sigma}(\mathbf{a}_{\sigma}-I)   )|, |\det(\mathbf{x}_{1,\sigma})|, \text{ and } |\det(\mathbf{x}_{2,\sigma})|$$
be uniformly bounded from below by $\frac{1}{c}$ for some $c>0$.  Finally, let the family of maps  $n_{\vec{k} ,\sigma}(\vec{K})\in \Bi(\Hi_{1},\Hi_{2}) $ satisfy the norm bound:
$$\sup_{\vec{k},\sigma}\| n_{\vec{k} ,\sigma}\| \leq r.$$
Then $B$ is well defined as a strong limit and is bounded in operator norm by
$$\|B\|\leq c r\|\eta\|_{1}.  $$

\end{proposition}

\begin{proof}
We check the conditions for Proposition~\ref{StrongConv} (applied for integrals rather than sums).  Due to the intertwining relations between the unitaries $\tau_{\vec{k}}^{*}\tau_{\mathbf{a}_{\sigma} \vec{k} }$ and the multiplication operators $q_{\vec{k} ,\sigma}$, we will then have a bound from above by an integral of multiplication operators.    We must show that $\frac{1}{2}(G_{1}+G_{2})$ is bounded, where
 \begin{align*}
  G_{1}= \int d\vec{k} \int_{ SO_{n} } d\sigma \,  |     \tau_{\vec{k}}^{*}\,\tau_{ \mathbf{a}_{\sigma}\vec{k}}\,q_{\vec{k}_{2} ,\sigma_{2}} | \text{ and } G_{2}= \int d\vec{k}  \int_{ SO_{n} }d\sigma \, |     q^{*}_{\vec{k}_{1} ,\sigma_{1}}\, \tau_{\mathbf{a}_{\sigma} \vec{k}}^{*}\,\tau_{\vec{k}}  |.
\end{align*}
The integrand of $G_{1}$ is the multiplication operator
\begin{align*}
|     \tau_{\vec{k}}^{*}\tau_{\mathbf{a} \vec{k}}q_{\vec{k} ,\sigma} |= |n_{\vec{k} ,\sigma}(\vec{P})| | \eta(\mathbf{x}_{1,\sigma}\vec{k}+\mathbf{y}_{\sigma}\vec{P} , \mathbf{x}_{2,\sigma}\vec{k}+\mathbf{y}_{\sigma}\vec{P}  )|.
\end{align*}
and the integrand of $G_{2}$ is
\begin{multline}
|    q^{*}_{\vec{k} ,\sigma}\,\tau_{\mathbf{a}_{\sigma} \vec{k}}^{*}\,\tau_{\vec{k}}  |= \tau_{\vec{k}}^{*}\, \tau_{\mathbf{a}_{\sigma} \vec{k}}^{*}\,|n_{\vec{k} ,\sigma}(\vec{P})| | \eta(\mathbf{x}_{1,\sigma}\vec{k}+\mathbf{y}_{\sigma}\vec{P} , \mathbf{x}_{2,\sigma}\vec{k}+\mathbf{y}_{\sigma}\vec{P}  )|\,\tau_{\mathbf{a}_{\sigma} \vec{k}}^{*}\,\tau_{\vec{k}}\\=|n_{\vec{k} ,\sigma}(\vec{P}+\sigma\vec{k}-\vec{k} )| | \eta(\mathbf{x}_{1,\sigma}^{\prime}\vec{k}+\mathbf{y}_{\sigma}\vec{P} , \mathbf{x}_{2,\sigma}^{\prime}\vec{k}+\mathbf{y}_{\sigma}\vec{P}   )|
\end{multline}
where $\mathbf{x}_{j,\sigma}^{\prime}=\mathbf{x}_{j,\sigma}+\mathbf{y}_{\sigma}(\mathbf{a}_{\sigma}-I)$, and  we have used that $\tau_{k}M(\vec{P})=M(\vec{P}-k)\tau_{k} $.

However since the operators in the integrand of $G_{1}$ are all multiplication operators in $\vec{P}$, bounding a sum on them in the operator norm can be computed as a supremum in the following way:
\begin{multline}\label{Lastish}
\| G_{1}\| \leq  \sup_{\vec{P} }\|   \int d\vec{k}\int_{ SO_{n} }d\sigma \,|n_{\vec{k} ,\sigma}(\vec{P})| | \eta(\mathbf{x}_{1,\sigma}\vec{k}+\mathbf{y}_{\sigma}\vec{P} , \mathbf{x}_{2,\sigma}\vec{k}+\mathbf{y}_{\sigma}\vec{P}  )|       \|_{\Bi(\Hi_{1})} \\ \leq \big(\sup_{\vec{P} }\| n_{\vec{k} ,\sigma}(\vec{P})\|_{\Bi(\Hi_{1})} \big)\sup_{\vec{P} }  \big( \int d\vec{k}\int_{ SO_{n} }d\sigma \, | \eta(\mathbf{x}_{1,\sigma}\vec{k}+\mathbf{y}_{\sigma}\vec{P} , \mathbf{x}_{2,\sigma}\vec{k}+\mathbf{y}_{\sigma}\vec{P}  )| \big)
\end{multline}

A similar result holds for $G_{2}$.   Now applying Lemma~\ref{DensityMat} to~(\ref{Lastish}) along with our conditions on $\mathbf{x}_{1,\sigma}$, $\mathbf{x}_{2,\sigma}$, and $n_{\vec{k} ,\sigma}(\vec{P})$ we get the bound  $\| G_{1}\|\leq rc\|\eta\|_{1}$.

\end{proof}

\begin{corollary}\label{IntSpec1}
The integral of operators~(\ref{IntForm1}) converges strongly to a bounded operator with norm less than or equal to $\frac{1}{(1-\lambda)^{n}}\|\rho\|_{1}$.
\end{corollary}
The bound in the above corollary in not sharp, since in Proposition~(\ref{GoodForm1}) we show that $\tilde{B}=\Tr_{2}[\rho \mathbf{A}_{\lambda}^{*}]$.   Thus $\|\tilde{B}\|\leq \|\rho\|_{1}\|\Sca_{\lambda}-I\|\leq 2\|\rho\|$, since $\Sca_{\lambda}$ is unitary.

\begin{proof}
We apply Proposition~(\ref{IntBnd1}) with $n_{\vec{k} ,\sigma}(\vec{P})=\scat_{\lambda}(|\vec{k}|)$, $\eta=\rho$, $\mathbf{a}_{\sigma}=\sigma$,  $\mathbf{x}_{1,\sigma}=1+\lambda$, $\mathbf{x}_{2,\sigma}=I+\sigma$, and $\mathbf{y}_{\sigma}=\lambda$.  $|n_{\vec{k} ,\sigma}(\vec{P})|\leq 1$, so we can take $r=1$.  All determinants involved are of operators of the form $\sigma_{1}+\lambda \sigma_{2}$ where $\sigma_{1},\sigma_{2}\in SO_{n}$, so these determinants have a lower bound of $(1-\lambda)^{n}$.   Hence we can take $c=(1-\lambda)^{-n}$.

\end{proof}

\begin{proposition}\label{IntBnd2}
Let $G\in \mathcal{B}(\Hi_{l}\otimes L^{2}(\R^{n}) ,\Hi_{r}\otimes L^{2}(\R^{n}) )$, and $\varphi:\mathcal{B}(\Hi_{l}\otimes L^{2}(\R^{n}) ,\Hi_{r}\otimes L^{2}(\R^{n}) )\rightarrow \mathcal{B}(\Hi_{l}^{0}\otimes L^{2}(\R^{n}) ,\Hi_{r}^{0}\otimes L^{2}(\R^{n}) )$ has the form
\begin{align*}
\varphi(G)=\sum_{j}\int d\vec{k}\int_{SO_{n}\times SO_{n}}d\sigma_{1}d\sigma_{2}\, U^{*}_{\vec{k},\sigma_{1}}\,h^{*}_{j,\vec{k},\sigma_{1}}\,G \,g_{j,\vec{k},\sigma_{2}}\,U_{\vec{k},\sigma_{2}},
\end{align*}
where $U_{\vec{k},\sigma}$ acts  on the $L^{2}(\R^{n})$ tensor as $
U_{\vec{k},\sigma}=  \tau_{\vec{k}}\, D_{\mathbf{b}_{\sigma}}\, \tau_{\mathbf{a}_{\sigma}\vec{k}}^{*}$, and $h_{j,\vec{k},\sigma}$ and $g_{j,\vec{k},\sigma}$ are multiplication operators in $\vec{P}$ of the form:
\begin{align*}
h_{j,\vec{k},\sigma}=n_{j,\vec{k},\sigma}^{(1)}(\vec{P})\eta_{j}^{(1)} ( \mathbf{x}_{1,\sigma}\vec{k}+ \mathbf{x}_{2,\sigma} \vec{P}), \text{ and }
g_{j,\vec{k},\sigma}=n_{j,\vec{k},\sigma}^{(2)}(\vec{P})\eta^{(2)}_{j}(\mathbf{x}_{1,\sigma}\vec{k}+ \mathbf{x}_{2,\sigma}\vec{P}).
\end{align*}
In the above, $\mathbf{x}_{1,\sigma},\mathbf{x}_{2,\sigma},\mathbf{a}_{\sigma}\in M_{n}(\R)$, $\mathbf{b}_{\sigma}\in GL_{n}(\R)$,   the family of operators $n_{j,\vec{k},\sigma}^{(1)}$ and  $n_{j,\vec{k},\sigma}^{(2)}$ lie in $\mathcal{B}(\Hi_{l}, \Hi_{l}^{0} )$ and $\mathcal{B}(\Hi_{r},\Hi_{r}^{0})$, respectively, and finally $\eta_{j}^{(1)}, \eta_{j}^{(2)}\in L^{2}(\R^{n})$.    We will require that
$$
\inf_{\sigma}|\det(\mathbf{x}_{1,\sigma}+\mathbf{x}_{2,\sigma}(\mathbf{b}_{\sigma}^{-1}\mathbf{a}_{\sigma}-I))|\geq \frac{1}{c}
$$
and $\sup_{j,\vec{k},\sigma} \|n_{j,\vec{k},\sigma}^{(1)}\|,\sup_{j,\vec{k},\sigma}\|n_{j,\vec{k},\sigma}^{(2)}\| \leq r$.

In this case, the integral of operator converges strongly to an operator $\varphi(G)$ with the norm bound
$$
\|\varphi(G)\| \leq   c r^2\frac{1}{2}(\|T_{1}\|_{1}+\|T_{2}\|_{1})\|G\|.
$$ 
where $T_{\epsilon}=\sum_{j}|\eta_{j}^{(\epsilon)}\rangle \langle \eta_{j}^{(\epsilon)}|$ for $\epsilon=1,2$ and $\|\cdot \|_{1}$ is the trace norm.

\end{proposition}
\begin{proof}
We work towards showing the conditions of Proposition~\ref{StrongConv2} with sums replaced an integral-sum.  We thus need to find bounds for the operator norms of 
\begin{multline}
  \sum_{j} \int_{SO_{n}}d\sigma\int_{\R^{n}}d\vec{k}\,   U_{\vec{k},\sigma}^{*}\,  |g|^{2}_{j,\vec{k},\sigma}\, U_{\vec{k},\sigma}\text{ and }   \sum_{j}  \int_{SO_{n}}d\sigma   \int_{\R^{n}}d\vec{k}\,U_{\vec{k},\sigma}^{*}\,|h|^{2}_{j,\vec{k},\sigma}\,U_{\vec{k},\sigma}.
\end{multline}
 $|g|^{2}_{j,\vec{k},\sigma}=|g|^{2}_{j,\vec{k},\sigma}(\vec{P})$ is a multiplication operator with elements in $\mathcal{B}(\Hi_{l}^{0},\Hi_{l}^{0})$  or an element in $\mathcal{B}(L^{2}(\R^{n})\otimes\Hi_{l},L^{2}(\R^{n})\otimes\Hi_{l})$.  Conjugating with   $U_{\sigma} $, we get only multiplication operators back:
$$
U_{\vec{k},\sigma}^{*}\,  |g|^{2}_{j,\vec{k},\sigma}(\vec{P})\,U_{\vec{k},\sigma}=|g|^{2}_{j,\vec{k},\sigma}( \mathbf{b}_{\sigma}^{-1}\vec{P}+ ( \mathbf{b}_{\sigma}^{-1}\mathbf{a}_{\sigma}-I)\vec{k} ).
$$
With the calculations for bounding integrals of multiplication operators as in the proof of~(\ref{IntBnd1}), we get the bound
$$
\|\varphi(G)\|\leq \frac{1}{2}c\, r^{2}\sum_{j}(\|\eta_{j}^{(1)}\|_{2}^{2}+\|\eta_{j}^{(2)}\|_{2}^{2})\|G\|=\frac{1}{2}c\, r^{2}(\|T_{1}\|_{1}+\|T_{2}\|_{1})\|G\|.
$$

\end{proof}

\begin{corollary}\label{IntSpec2}
The integral of operators~(\ref{FormHere}) converges strongly to a limit with operator norm bounded by $\|\rho\|_{1}\|G\|\big(\frac{1}{1-\lambda}\big)^{n}$.
\end{corollary}

\begin{proof}
We apply Proposition~\ref{IntBnd2}  in the case where $\mathbf{a}_{\sigma}=\sigma$, $\mathbf{b}_{\sigma}=\frac{I+\lambda \sigma}{1+\lambda}$, $\mathbf{x}_{1,\sigma}=\frac{\sigma+\lambda}{1+\lambda}$,  $\mathbf{x}_{2,\sigma}=\frac{\lambda(I-\sigma)}{I+\lambda }$, $\eta_{j}^{(1)}=\eta_{j}^{(2)}=\sqrt{\beta_{j}}f_{j}$, and
$$n_{j,\vec{k},\sigma}^{(1)}(\vec{P})=n_{j,\vec{k},\sigma}^{(2)}(\vec{P})= \det(1+\lambda \sigma )^{-\frac{1}{2}} .$$
In this case $|n_{j,\vec{k},\sigma}^{(1)}(\vec{P})|\text{ and } |n_{j,\vec{k},\sigma}^{(2)}(\vec{P})|\leq (1-\lambda)^{-\frac{n}{2}} $, so we can take $r=1$.    Also $\mathbf{x}_{1,\sigma}+\mathbf{x}_{2,\sigma}(\mathbf{b}_{\sigma}^{-1}\mathbf{a}_{\sigma}-I)= \frac{\sigma(1+\lambda)}{I+\lambda \sigma} $ and $\|(\frac{\sigma(1+\lambda)}{I+\lambda \sigma} )^{-1}\|\leq1$, and hence $|\det(\frac{\sigma(1+\lambda)}{I+\lambda \sigma} )|\geq (1)^{n}=1$ independent of $\lambda$ and $\sigma$, so we can take $c=1$.   Hence by~(\ref{IntBnd2}), we have our conclusion with a bound $\|\rho\|_{1}\|G\|(1-\lambda)^{-n} $.

\end{proof}

\section{Reduced Born approximation with third-order error}\label{Approx}
In this section, we will prove Theorem~\ref{HardLem12}.   To make mathematical expression more compact it will be helpful to have the dictionary below.   In the following expressions $\lambda, r\in \R^{+}$, $\sigma\in SO_{n}$, and $\vec{k},\vec{P}\in \R^{n}$.

\vspace{1cm}
\noindent \textbf{\large{Dictionary of vectors in $\R^{n}$\hspace{2cm} Dictionary of matrices in $M_{n}(\R)$}}

\begin{enumerate}
\item $\art=(1+r\lambda)\vec{k}+r\lambda \vec{P}$   \hspace{3.1cm} $1.$ $\clt= \frac{\sigma(1+\lambda) }{ \sigma(1+\lambda )-\lambda r (\sigma-I) }  $

\item $\vt= \frac{\sigma+\lambda}{1+\lambda}\vec{k}-\lambda \frac{\sigma-I}{1+\lambda}\vec{P}$\hspace{3.7cm}$2.$   $\crt=\frac{-\lambda r (1-\lambda)  }{ \sigma(1+\lambda )-\lambda r(\sigma -I) } $

\item
$\vrt= \frac{\sigma(1 +\lambda  )-\lambda r(\sigma-I) }{1+\lambda }\vec{k}-\lambda r \frac{\sigma-I}{1+\lambda } \vec{P}$ \hspace{1.4cm}$3.$ $\cut= \frac{ (1+\lambda )(r+(1-r)\sigma)   }{ \sigma(1+\lambda )-\lambda r (\sigma-I)  }$

\item $ \kt =\frac{1}{1+\lambda}\vec{k}-\frac{\lambda}{1+\lambda} \vec{P} $

\end{enumerate}

Now we will list  some relations between the vectors.  The significance of these relations will become apparent once we begin doing calculations.\\

\noindent \textbf{\large{Relations}}\\
\begin{enumerate}

\item[R1.]   $\vec{k}+\vec{P}= \frac{1}{1+r\lambda}\art+\frac{1}{1+r\lambda}\vec{P}$

\item[R2.]   $ \kt = \clt  \vrt-\lambda \cut \vec{P} $

\item[R3.] $ \vec{k}+\vec{P}  = \clt \vrt + \crt \vec{P}     $

\end{enumerate}

In the proof of~(\ref{HardLem12}) the analysis is organized around the fact that certain expressions are bounded.   In the limit $\lambda\rightarrow 0$, expressions of the type $\frac{1}{\lambda}\bar{\scat}_{\lambda}(\cdot)$ will be a source of unboundedness, and  $\rho$ and $G$ will have to be constrained in such a way as to compensate for this.    The following expressions, defined for dimensions $n=1,3$, are uniformly bounded in $P,k\in \R$, $\sigma\in \{+,-\}$, $0 \leq r\leq 1$, \text{and} $0\leq\lambda $:
\begin{eqnarray}\label{BoundedOnes1}
 E_{1}(\vec{P},\vec{k},r,\lambda)&=&\frac{ (\delta_{n,3}+|\art|^{n-2})^{-1}  }{1+|\vec{P}|} \frac{1}{\lambda}\bar{\scat}_{\lambda}(|\vec{k}|),\\
 \label{BoundedOnes2} E_{2}(\vec{P},\vec{k},\sigma,r,\lambda)&=&\frac{(\delta_{n,3}+|\vrt|^{n-2})^{-1} }{1+ |\vec{P}|}\frac{1}{\lambda}\bar{\scat}_{\lambda}(|\kt|),\\
\label{BoundedOnes3} E_{3}(\vec{P},\vec{k},\lambda)&=& \frac{(\delta_{n,3}+|\vec{k}|^{n-2})^{-1} }{1+ |\vec{P}|}\frac{1}{\lambda}\bar{\scat}_{\lambda}(|\kt |).
\end{eqnarray}
Their boundedness can be seen by using  (R1) to rewrite $\vec{k}$ in terms of $\arto$ and $\vec{P}$ for $E_{1}(\vec{P},\vec{k},r,\lambda)$,  (R2) to write $\kt$ in terms of $\vro$ and $\vec{P}$ for $E_{2}(\vec{P},\vec{k},\sigma,r, \lambda)$, and for $E_{3}(\vec{P},\vec{k},\lambda)$, $\kt$ explicitly defined in terms of $\vec{k}$ and $\vec{P}$.

 A second-order Taylor expansion of the scattering coefficients gives:
\begin{enumerate}
 \item[ \hspace{1cm} \textbf{Dim}-1]
\begin{align*}
\scat_{\lambda} (k)=\frac{-i \alpha \frac{\lambda}{1+\lambda} }{k+i\frac{1}{2} \alpha\frac{\lambda}{1+\lambda} }\sim -\lambda(1-\lambda) \frac{i\alpha}{k}-\frac{\lambda^2}{2}\frac{\alpha^{2}}{k^{2}}
\end{align*}
\item[\textbf{Dim}-2]
 \begin{align}\label{ExpScat2}
 \scat_{\lambda  }(k)=   \frac{-i\pi}{ \frac{1+\lambda}{\lambda} \mathit{l}^{-1}+\gamma+\ln(\frac{k}{2})-i \frac{\pi}{2} }\sim -\lambda(1-\lambda) i \pi \mathit{l}-i\lambda^{2}\mathit{l}^{2}(\gamma+\ln(\frac{k}{2}))-\frac{\lambda^{2}}{2} \pi 
 \end{align}
\item[\textbf{Dim}-3]
 \begin{align*}
 \scat_{\lambda }(k)=\frac{-2ik}{ \frac{1+\lambda}{\lambda}\mathit{l}^{-1} +ik}\sim   -\lambda(1-\lambda) 2i\mathit{l}k-2\lambda^2\mathit{l}^{2}k^{2}
 \end{align*}
 \end{enumerate}
We can summarize the above expressions as
$$\scat_{\lambda}(k)\sim -i\lambda(1-\lambda)\mathit{c}_{n}k^{n-2}-\frac{\lambda^{2}}{2}\mathit{c}_{n}^{2}k^{2(n-2)}-\delta_{n,2}i\lambda^{2}\mathit{l}^{2}(\gamma+\ln(\frac{k}{2})),$$
where $\mathit{c}_{\mathbf{1}}=\alpha$, $\mathit{c}_{\mathbf{2}}= \pi \mathit{l}$, and $  \mathit{c}_{\mathbf{3}}= 2\mathit{l}$.  We will use the following simple lemma.

\begin{lemma}\label{USE} Let $k,K\in \R$ and $\vec{k},\vec{K}\in \R^{3}$.

\begin{enumerate}
\item We have the inequality $$\frac{1}{\sqrt{(k-\lambda K)^2+\frac{\alpha^{2}}{4}\lambda ^2}}\leq 2\frac{\sqrt{K^2+\frac{\alpha^{2}}{4}}}{\alpha |k|}\leq \frac{2|K|+\alpha }{\alpha |k|},$$
\item  and for dimension one the scattering coefficient satisfies
$$
\big|\scat_{\lambda}\big(\big| \frac{k-\lambda K}{1+\lambda}\big| \big)\big|\leq \lambda \frac {2|K|+\alpha }{|k|},
$$
\item and
$$
\big| \scat_{\lambda}\big(\big|\frac{k-\lambda K}{1+\lambda}\big|\big)-\frac{-i\alpha \lambda}{|k|}\big|\leq \lambda^2 |K|\frac{2|K|+\alpha }{|k|^{2}}.
$$
\item
for dimension three, the scattering coefficient satisfies
$$
\big| \scat_{\lambda}\big(\big|\frac{\vec{k}-\lambda \vec{K}}{1+\lambda}\big|\big)-(-2i\mathit{l} \lambda |\vec{k}|) \big|  \leq \lambda^{2}\frac{4\mathit{l}}{(1+\lambda)^{2}}(1+\mathit{l}|\vec{k}|)(|\vec{k}|+ |\vec{K}|).
$$

\end{enumerate}

\end{lemma}
\begin{proof} (1) follows by evaluating the critical points in $\lambda$.  (2) and (3) follow with an application of (1).

\end{proof}

Define the following weighted trace norm $\|\cdot \|_{wtn}$ for the density matrices on the single reservoir particle Hilbert space $\rho$:
\begin{multline}
\|\rho \|_{wtn}=\|\rho\|_{1} +\sum_{\epsilon}\sum_{1\leq i,j \leq n}\| |\vec{P}|^{n-2+\epsilon} [X_{i},[X_{j},\rho]]\|_{1}\\+\sum_{\epsilon}\sum_{j=1}^{n}\| |\vec{P}|^{n-2+\epsilon} X_{j}\rho X_{j}|\vec{P}|^{n-2+\epsilon}\|_{1}+\||\vec{P}|^{2(n-2)}\rho |\vec{P}|^{2(n-2)} \|_{1},
\end{multline}
where the sums in $\epsilon$ are over $\{0,1\}$ for dimension one and $\{-1,0,1\}$ for  dimension three.   Notice the contrast between dimension $n=1$ and $n=3$ with respect to the weights applied in the norms for the absolute value of the momentum operators $|\vec{P}|$.   For $n=1$, $\|\rho\|_{wtn}$ will blow up if $\rho$ has non-zero density of momenta near momentum zero, while for $n=3$, $\|\rho\|_{wtn}$ can blow up if the momentum density does not decay fast enough for large momenta.    This difference in requirements for different dimensions can be seen also in the formulas~(\ref{one}-\ref{four}). The norm $\|\rho\|_{wtn}$ is not really asymmetric with respect to operators multiplying from the left and the right when $\rho$ is self-adjoint.

\begin{theorem}\label{HardLem12}
Let $\epsilon(G,\lambda)$ be defined as in~(\ref{PoissonProc}), then there exists a $c$ s.t. for all  density operators $\rho\in \Bi_{1}(L^{2}(\R^{n}))$, $G\in \Bi(L^{2}(\R^{n}))$, and $0\leq \lambda$
\begin{align}\label{ERROR}
\|\epsilon(G,\lambda)\|\leq c\lambda^{3}\|\rho\|_{wtn}  \|G\|_{wn }.
\end{align}

\end{theorem}

\begin{proof}
We will prove the result for density operators $\rho$ with a twice continuously differentiable integral kernel $\rho(\vec{k}_{1},\vec{k}_{2})$ in the momentum representation, and a spectral decomposition $\rho=\sum_{j=1}^{\infty}\lambda_{j} |f_{j}\rangle\langle f_{j}|$  of vectors $f_{j}(\vec{k})$ that are continuously differentiable in the momentum representation.  Since such $\rho$ are dense  with respect to the $\|\cdot \|_{wtn}$, the result extends to all $\rho$ with  $\| \rho\|_{wtn}<\infty$.   By~(\ref{LimitExp}), the  $V_{1}$, $V_{2}$, $\vec{A}$ operators and the map $\varphi$ are well defined for all $\rho$ with $\|\rho\|_{wtn}<\infty$ and they vary continuously as a function of $\rho$ with respect to the norm $\|\cdot \|_{wtn}$. 

Our challenge is to expand the expressions we found in Propositions~\ref{GoodForm1} and~\ref{GoodForm2} in $\lambda$, until we reach our second-order Taylor expansion while making sure that we only throw away terms which are bounded as in~(\ref{ERROR}).  We will organize our analysis  using the expressions~(\ref{BoundedOnes1}),~(\ref{BoundedOnes2}), and~(\ref{BoundedOnes3}), in conjunction with Propositions~\ref{IntBnd1} and~\ref{IntBnd2} to effectively transfer the conditions for the boundedness of the differences in our expansions to conditions on $\rho$ and $G$.   Both of the expressions~(\ref{IntForm1}) and~(\ref{FormHere}) have multiple sources of $\lambda$ dependence.    If we expand the expressions involving $\rho$ and $f_{j}$ first for~(\ref{IntForm1}) and~(\ref{FormHere}) respectively, then the resulting expressions left to expand will be summable in the operator norm and thus not require the heavy preparation involved with the use of Propositions~\ref{IntBnd1} and~\ref{IntBnd2}.    Breaking $\Tr_{2}[\rho\Sca_{\lambda}^{*}G\Sca_{\lambda}]$ into parts and dividing by $\lambda$ we just need bound the differences
  \begin{align}\label{label1}
  \frac{1}{\lambda}\Tr_{2}[\rho\Asca^{*}_{\lambda}]G-(iV_{1}+i\lambda V _{2}+i\frac{\lambda}{2}\{\vec{A},\vec{P}\}-\frac{\lambda}{2} \varphi(I))G \text{, and}   
  \end{align}
  \begin{align}\label{label2}
   \frac{1}{\lambda} \Tr_{2}[\rho\Asca_{\lambda}^{*}G\Asca]-\lambda \varphi(G),
    \end{align}
where there is a similar expression to~(\ref{label1}) for $\frac{1}{\lambda}G\Tr_{2}[\rho \Asca_{\lambda}]$.  
 We begin with~(\ref{label1}), and will have to bound a sequence of  intermediate differences.  The main differences are the following:\\

\noindent\textbf{Difference $\mathbf{1}$}\\
\begin{align*}
\| \frac{1}{\lambda}\Tr_{2}[\rho\Asca^{*}_{\lambda}]G-\int_{\R^{n}}d\vec{k} \int_{SO_{n} }d\sigma \, \tau_{\vec{k}}\, \tau_{\sigma \vec{k}}^{*}\, \big(\rho(\vec{k},\sigma \vec{k})+\lambda(\vec{P}+ \vec{k})\nabla_{T}\rho(\vec{k},\sigma \vec{k}) \big)  \frac{1}{\lambda}\bar{\scat}_{\lambda}(|\vec{k}|)G\|,
\end{align*}
\textbf{Difference $\mathbf{2}$}\\
\begin{multline*}
\|\int_{\R^{n}}d\vec{k}\int_{SO_{n}} d\sigma \, \tau_{k} \,\tau_{\sigma \vec{k} }^{*}\, (\rho(\vec{k},\sigma \vec{k})+\lambda \big(\vec{P} + \vec{k})\nabla_{T}\rho(\vec{k},\sigma \vec{k}) \big)\\   \big(\frac{1 }{\lambda}\bar{\scat}_{\lambda}(|\vec{k}|)-\big( (1-\lambda)\mathit{c}_{n} |\vec{k}|^{2-n} +\frac{\lambda}{2} \mathit{c}_{n}^{2} |\vec{k}|^{2(2-n)}    \big)\big)G\|,
\end{multline*}
\textbf{Difference $\mathbf{3}$}\\
\begin{multline*}
\|\int_{\R^{n}}d\vec{k} \int_{SO_{n} }d\sigma\,\tau_{\vec{k}}\, \tau_{\sigma \vec{k}}^{*}\, \big(\rho(\vec{k},\sigma \vec{k})+ \lambda (\vec{P}+ \vec{k})\nabla_{T}\rho(\vec{k},\sigma \vec{k}) \big) \\    \big( (1-\lambda) \mathit{c}_{n} |\vec{k}|^{2-n} +\frac{\lambda}{2}  \mathit{c}_{n}|\vec{k}|^{2(2-n)}    \big)G- (iV_{1}+i\lambda V_{2}+\frac{\lambda}{2}\{\vec{A},\vec{P}\}-\frac{\lambda}{2} \varphi(I)     )G\|.
\end{multline*}

By the differentiability properties of the integral kernel $\rho$,
\begin{multline*}
\rho(\vec{k}+\lambda(\vec{P}+ \vec{k}),\sigma \vec{k} +\lambda (\vec{P}+ \vec{k}))= \rho(\vec{k},\sigma \vec{k})+\lambda (\vec{P}+\vec{k})\nabla_{T}\rho(\vec{k},\sigma \vec{k})\\ +\lambda^2(\vec{P}+ \vec{k})^{\otimes^{2}}\int_{0}^{1}ds \int_{0}^{s}dr\,\nabla_{T}^{\otimes^{2}} \rho(\vec{k}+\lambda (\vec{P}+ \vec{k})r,\sigma \vec{k}+\lambda (\vec{P}+ \vec{k})r ),
\end{multline*}
where $\nabla_{T}^{\otimes^{2}}g(x,y)$ is $2$ tensor of derivatives with
$$(\nabla_{T}^{\otimes^{2}}g(x,y))_{(i,j)}=\lim_{h\rightarrow 0}\frac{(\nabla_{T}g)_{i}(x+he_{j},y+he_{j})-(\nabla_{T}g)_{i}(x,y)}{h}.$$

The first difference can be rewritten as
\begin{multline*}
\lambda^{2}\int_{0}^{1}  ds  \int_{0}^{s}dr\, \| \int_{\R^{n}}d\vec{k} \int_{SO_{n}}d\sigma \,\tau_{\vec{k}} \,\tau_{\sigma \vec{k}}^{*}\, (1+|\vec{P}|)\,(\vec{P}+ \vec{k})^{\otimes^{2}}\, (\delta_{n,3}+|\arto|^{n-2}) \\  \nabla_{T}^{\otimes^{2}}\rho(\arto,  \yrto ) )  E_{1}(\vec{P},\vec{k},r,\lambda)G\|.
\end{multline*}
Using (R1) and expanding the tensor:  $(\arto+\vec{P})^{\otimes^{2}} $  a single term has the form $\arto^{\otimes^{m}}\vec{P}^{\otimes^{2-m}}$.  Note that the order of the tensors does not matter in this situation, since the whole vector is in an inner product with $\nabla^{\otimes^2}\rho$, and partial derivatives commute.  Now we apply Proposition~\ref{IntBnd1} with a single term:
\begin{eqnarray*}
n_{\vec{k},\sigma}&=&  E_{1}(\vec{P},\vec{k},r,\lambda)(\frac{1}{1+r\lambda \sigma})^{2}\frac{\big(\arto^{\otimes^{m}}\otimes \vec{P}^{\otimes^{2-m}}\big)_{j,k} }{| \arto|^{m}|\vec{P}|^{2-m} },
\\
\eta &=&  (\delta_{n,3}+|\vec{k}|^{n-2})|\vec{k}|^{m} (\nabla_{T}^{\otimes^{2}}\rho )_{j,k},
\\
q_{\vec{k},\sigma}&=&n_{\vec{k},\sigma,\lambda}\eta(\arto,\yrto).
\end{eqnarray*}
Finally with~(\ref{IntBnd1}) we get the bound $\lambda^{2}C\|(\delta_{n,3}+ |\vec{P}|^{n-2}) |\vec{P}|^{m}(\nabla_{T}^{\otimes^{2}}\rho)_{j,k}\|_{1} \||\vec{P}|^{2-m}(I+|\vec{P}|)G\|$, for some constant  $C$.   Note that $  \nabla_{T}\rho=i(\vec{X}\rho-\rho \vec{X})$.

The second difference can be bounded for dimension-one using the inequality
$$
|\frac{1}{\lambda }\bar{\scat}_{\lambda}(|\vec{k}|)-(\frac{i\alpha (1-\lambda)}{|\vec{k}|}-\frac{\lambda \alpha^{2}}{2|\vec{k}|^2})| \leq \frac{\lambda^2\alpha^{3} }{|\vec{k}|^3},
$$
and for dimension three using the inequality
$$
|\frac{1}{\lambda \el } \scat_{\lambda}(|\vec{k}|)-( 2i (1-\lambda)|\vec{k}|-2\lambda  \el |\vec{k}|^2)|\leq 2\lambda^2\el^{2} |\vec{k}|^{3}.
$$

Finally, the last difference comes down to bounding the cross term:
\begin{align*}
\|\int_{\R^{n}}dk \int_{SO_{n} }d\sigma \,\tau_{\vec{k}}\, \tau_{\sigma \vec{k}}^{*}\, \lambda(\vec{P}+\sigma \vec{k})\nabla_{T}\rho(\vec{k},\sigma \vec{k})\,( \lambda^{2} \mathit{c}_{n}|\vec{k}|^{n-2}+\lambda^2  \mathit{c}_{n}^{2}\, |\vec{k}|^{2(n-2)})\, G\|.
\end{align*}
The bound for the above term follows from~(\ref{DensityMat}) and that $\int d\vec{k}\, \rho(\vec{k},\vec{k}) |\vec{k}|^{2(n-2)}=\| |\vec{P}|^{n-2} \rho |\vec{P}|^{n-2}\|_{1}$.

  The $\frac{1}{\lambda}G\Tr_{2}[\rho\Asca_{\lambda}] $ is similarly analyzed so now we study~(\ref{label2}).  Again we have three main differences.   There is a $\lambda$ dependence in $m_{j,\vec{k},\sigma,\lambda}$, $U_{\vec{k},\sigma,\lambda}$, and $\scat_{\lambda}(| \kt | )$.   It is most convenient to begin expanding $m_{j,k,\sigma,\lambda}$ first.

   \vspace{.5cm}

\noindent \textbf{Difference $\mathbf{1}$ }\\
\begin{multline*}
\| \frac{1}{\lambda}\Tr_{2}[\rho \Asca_{\lambda}^{*}G\Asca_{\lambda}]-  \sum_{j} \frac{1}{\lambda}  \int_{\R^{n}}d\vec{k} \int_{SO_{n}\times SO_{n}}d\sigma_{1}\,d\sigma_{2}\, U^{*}_{\vec{k},\sigma_{1},\lambda}\,\det(I+\lambda \sigma_{1})^{-\frac{1}{2}}\, \bar{f}_{j}(\vec{k}) \\    \bar{\scat}_{\lambda}(| \kt| )\,G\,   \scat_{\lambda}(| \kt | )\, f_{j}(\vec{k})\, \det(I+\lambda \sigma_{2})^{-\frac{1}{2}}\, U_{\vec{k},\sigma_{2},\lambda}    \|.
\end{multline*}
\noindent \textbf{Difference $\mathbf{2}$ }\\
\begin{multline*}
\|\sum_{j} \int_{ \R^{n}}d\vec{k} \int_{SO_{n}\times SO_{n} }d\sigma_{1}\,d\sigma_{2}\, U^{*}_{\vec{k},\sigma_{1},\lambda}\, \bar{f}_{j}(\sigma_{1}\vec{k})\,\Big[ \lambda^{-1} \det(1+\lambda \sigma_{1})^{-\frac{1}{2}} \bar{\scat}_{\lambda}( | \kt |)\,G \\   \scat_{\lambda}(| \kt | )\,\det(1+\lambda \sigma_{2})^{-\frac{1}{2}} -\mathit{c}_{n}^{2}\,|\vec{k}|^{2(n-2)}  \,G \Big]     f_{j}(\sigma_{2}\vec{k})\,U_{\vec{k},\sigma_{2},\lambda }\|.
\end{multline*}
\noindent \textbf{Difference $\mathbf{3}$ }\\
\begin{align*}
\|\sum_{j} \lambda^2 \mathit{c}_{n}^{2} \int_{ \R^{3} }d\vec{k}\, |f_{j}(\vec{k})|^2|\vec{k}|^{2(n-2)}\int_{SO_{n}\times SO_{n} }d\sigma_{1}\,d\sigma_{2}\,(U^{*}_{\vec{k},\sigma_{1},\lambda}\,G \,U_{\vec{k},\sigma_{2},\lambda}-\tau^{*}_{\sigma_{1} \vec{k}}\,\tau_{\vec{k}}\,G\, \tau_{\vec{k}}^{*}\,\tau_{\sigma_{2}\vec{k} })\|.
\end{align*}

Using the differentiability of $f_{j}$'s
\begin{multline}\label{Expand1}
f_{j}\big(\sigma_{1} \vec{k} -\lambda  \frac{\sigma_{1}-1}{1+\lambda }( \vec{k}+\vec{P})\big)= f_{j}\big(\sigma_{1} \vec{k}\big)+\\ \lambda \big( \frac{\sigma_{1}-1}{1+\lambda} \big)(\vec{k}+\vec{P})\int_{0}^{1}dr\nabla f_{j}\big(\sigma_{1}\vec{k}+r\lambda \big( \frac{\sigma_{1}-1}{1+\lambda }\big)(\vec{k}+\vec{P}) \big).
\end{multline}
The first difference
\begin{multline*}
\| \frac{1}{\lambda}\Tr_{2}[\rho \Asca_{\lambda}^{*}G\Asca_{\lambda}]-\sum_{j}  \int_{ \R^{n} }d\vec{k} \int_{SO_{n}\times SO_{n} }d\sigma_{1}\,d\sigma_{2}\, U^{*}_{\vec{k},\sigma_{1},\lambda }\,  \det( 1+\lambda \sigma_{1})^{-\frac{1}{2}}\,  \bar{f}_{j}(\vec{k})\\    \bar{\scat}_{\lambda}(| \kt | )\, G \,   \scat_{\lambda}(| \kt| )\motj U_{\vec{k},\sigma_{2},\lambda} \|
\end{multline*}
is less than
\begin{multline*}
\lambda^2\mathit{c}_{n}^{2}  \int_{0}^{r} dr\| \sum_{j} \int_{ \R^{n}}d\vec{k} \int_{SO_{n}\times SO_{n} }d\sigma_{1}\,d\sigma_{2}\,U^{*}_{\vec{k},\sigma_{1},\lambda}\, \det( 1+\lambda \sigma_{1})^{-\frac{1}{2}}\, \\ \big(\frac{\sigma_{1}-I}{I+\lambda }\big)(\clto \vrto + \crto \vec{P})   )   \frac{\nabla\bar{f}_{j}(\vroo)}{ (\delta_{n,3}+|\vrto|^{n-2})^{-1} }\, \\ E_{2}(\vec{P},\vec{k},\sigma_{1},r,\lambda)(1+|\vec{P}|)\, G\, (1+|\vec{P}|) E_{2}(\vec{P},\vec{k},\sigma_{2},r,\lambda)\,  \frac{\mttj}{(\delta_{n,3}+ |\vrtt|^{n-2})^{-1}}\, U_{\vec{k},\sigma_{2},\lambda} \|,
\end{multline*}
where we have rearranged to substitute in the  $E_{2}(\vec{P},\vec{k},\sigma,r,\lambda)$ expressions and used (R3) to rewrite $\vec{k}+\vec{P}$.   Two applications of Proposition~\ref{IntBnd2}  corresponding to $\clto \vrt $ and $\crto \vec{P}$ will give us our bound.  For the $\clto \vrt  $ we use Proposition~(\ref{IntBnd2}) with
\begin{eqnarray*}
\eta^{(1)}_{j}(\vec{k})&=&(\delta_{n,3}+|\vec{k}|^{n-2})|\vec{k}| |\nabla f_{j}(\vec{k})|,
\\
n_{j,\vec{k},\sigma_{1}}^{(1)}(\vec{P})&=&\det(I+\lambda \sigma_{1})^{-\frac{1}{2}}\, E_{2}(\vec{P},\vec{k},\sigma,r,\lambda)\, \clto  \, \frac{\nabla f_{j}(\vec{k})}{|\nabla f_{j}(\vec{k})|},
\\
h_{j,\vec{k},\sigma_{1}}&=&n_{j,\vec{k},\sigma_{1}}^{(1)}\eta^{(1)}_{j}(\vto),
\\
\eta^{(2)}_{j}(\vec{k})&=& (\delta_{n,3}+|\vec{k}|^{n-2}) f_{j}(\vec{k}),
\\
n_{j,\vec{k},\sigma_{2}}^{(2)}(\vec{P})&=& \det(I+\lambda \sigma_{2})^{-\frac{1}{2}}\,  E_{2}(\vec{P},\vec{k},\sigma_{2},r,\lambda),
\\
g_{j,\vec{k},\sigma_{2}}&=&n_{j,\vec{k},\sigma_{2}}^{(2)}\eta^{(2)}_{j}(\vot).
\end{eqnarray*}
Hence the term is bounded by a constant multiple of
$$
\lambda^{2}\big(\sum_{j}\| |\vec{P}|(\delta_{n,3}+|\vec{P}|^{n-2})  X_{j}\rho X_{j}(\delta_{n,3}+|\vec{P}|^{n-2})|\vec{P}| \|_{1}  +  \| (\delta_{n,3}+|\vec{P}|^{n-2}) \rho(\delta_{n,3}+|\vec{P}|^{n-2})\|_{1}\big)\|(1+|\vec{P}|)G(1+| \vec{P}|)\|.
$$
The $\croo P$ term is bounded by a constant multiple of
$$
\lambda^{2}\big(\sum_{j}\|(\delta_{n,3}+|\vec{P}|^{n-2})X_{j}\rho X_{j} (\delta_{n,3}+|\vec{P}|^{n-2}) \|_{1}+ \| (\delta_{n,3}+|\vec{P}|^{n-2}) \rho(\delta_{n,3}+|\vec{P}|^{n-2})\|_{1}\big)\|(1+|\vec{P}|)|\vec{P}|G(1+| \vec{P}|)\|.
$$

The next intermediary difference has the form:
\begin{multline*}
\|\sum_{j} \int_{ \R^{n}}dk \int_{SO_{n}\times SO_{n} }d\sigma_{1}\,d\sigma_{2}\, U^{*}_{\vec{k},\lambda,\sigma_{1}}\, \det(1+\lambda \sigma_{1})^{-\frac{1}{2}}\, \bar{f}_{j}(\vec{k})\, \bar{\scat}_{\lambda}(| \kto| ) \\   G  \,(f( \vtt) - f(\sigma_{2} \vec{k}))\, \bar{\scat}_{\lambda}(| \kt | )\det(1+\lambda \sigma_{2})^{-\frac{1}{2}} \,U_{\vec{k},\lambda,\sigma_{1}}   \|.
\end{multline*}
Expanding $f( \vtt)- f(\sigma_{2} \vec{k}) $ as in~(\ref{Expand1}), we can apply a similar analysis to the above, except that for the left-hand side we organize around $E_{3}(\vec{P},\vec{k},\lambda)$ rather than $E_{3}(\vec{P},\vec{k},\sigma,r,\lambda)$.

 Due to $\bar{f}_{j}(\sigma_{1}\vec{k})f_{j}(\sigma_{2}\vec{k})$, the second difference is summable, and  we do not need to prepare any more applications of Proposition~\ref{IntBnd2}.    We begin by bounding
\begin{multline*}
\|\sum_{j} \lambda f \mathit{c}_{n}^{2} \int_{ \R^{n}}d\vec{k}  \int_{SO_{n}\times SO_{n}} d\sigma_{1}\,d\sigma_{2}\,U^{*}_{\vec{k},\sigma_{1},\lambda}\det(1+\lambda \sigma_{1})^{-\frac{1}{2}} \, \frac{\bar{f}_{j}(\sigma_{1}\vec{k}) }{|\vec{k}|^{2-n} } \,  \Big[\frac{|k|^{2-n} }{\lambda \mathit{c}_{n} }\, \bar{\scat}_{\lambda}( |\kto|)G \\    \frac{|k|^{2-n} }{\lambda \mathit{c}_{n} }\,\scat_{\lambda }(|\ktt |)   - G\Big]    \, \frac{ f_{j}(\sigma_{2}\vec{k})}{|\vec{k}|^{2-n} }\det(1+\lambda\sigma_{2})^{-\frac{1}{2}}\,  U_{\vec{k},\sigma_{2},\lambda}\|.
\end{multline*}
We observe the inequality
\begin{multline*}
\big\| \frac{|\vec{k}|^{2-n} }{\lambda \mathit{c}_{n} }\bar{\scat}_{\lambda}( |\kto |)G  \frac{|\vec{k}|^{2-n}}{\lambda \mathit{c}_{n} } \scat_{\lambda}(|\ktt |)- G]\big\| \leq  \\  \frac{1}{\mathit{c}_{n}}\big\|\big(\frac{|\vec{k}|^{2-n} }{\lambda \mathit{c}_{n} }\bar{\scat}_{\lambda}( |\kto| )-i\big) G(|\vec{P}|+I)  E_{3}(\vec{P},\vec{k},\lambda)\big\| \\  +\big\| G\big(\frac{|\vec{k}|^{2-n} }{\lambda \mathit{c}_{n}} \scat_{\lambda}(| \ktt |)+i \big)\big\|.
\end{multline*}
By (3) and (4) of~(\ref{USE}), the right-hand side is bounded by a sum of terms proportional to $\lambda |\vec{k}|^{r(n-2)}\| |\vec{P}|^{\epsilon_{1}} G (I+|\vec{P}|)^{\epsilon_{2}}  \|   $   for  $r=0,1,2$, $\epsilon_{1},\epsilon_{2}= 0,1$.   Bounding the above integral is then routine and requires that  $\| |\vec{P}|^{2(n-2)}\rho |\vec{P}|^{2(n-2)}\|_{1}$.   The last thing to do for the second difference is  expanding  $\det(1+\lambda \sigma_{1})^{-\frac{1}{2}}$ and  $|\det(1+\lambda \sigma_{2})|^{-\frac{1}{2}}$, which does not pose much difficulty.

For the third difference, we will need to work with the $D_{\frac{1+\lambda \sigma}{1+\lambda} }$ term.
$$\|U^{*}_{\vec{k},\sigma_{1},\lambda}\, G \,U_{\vec{k},\sigma_{2},\lambda}-\tau_{\sigma_{1}\vec{k}}^{*}\,\tau_{\vec{k}}\,G\, \tau_{\vec{k}}^{*}\,\tau_{\sigma_{2}\vec{k}} \| \leq \| (D_{\frac{1+\lambda \sigma_{1}}{1+\lambda} }^{*}-I) G\|+\|G (D_{\frac{1+\lambda \sigma_{1}}{1+\lambda} }-I)\|,$$
since $ U_{\vec{k},\sigma_{2},\lambda}$, $\tau_{\sigma \vec{k}}$, and $\tau_{\vec{k}}$ are unitary and $ D_{\frac{1+\lambda \sigma_{1}}{1+\lambda} }^{*}\tau_{\vec{k}}=  \tau_{\frac{1+\lambda\sigma}{1+\lambda }k} D_{\frac{1+\lambda \sigma_{1}}{1+\lambda} }^{*}$.
 $D_{\frac{1+\lambda \sigma}{1+\lambda}}$ satisfies the integral relation
$$
D_{\frac{1+\lambda \sigma}{1+\lambda}}=I+\int_{0}^{\lambda}ds\, \{ \frac{d}{ds}\log(\frac{1+s\sigma}{1+s})D_{\frac{1+s \sigma}{1+s}}\vec{P},\vec{X} \},
$$
and hence
$$\|\frac{1}{\lambda}(D_{\frac{1+\lambda \sigma}{1+\lambda}}-I)G  \|\leq (\sup_{0\leq s\leq \lambda}\frac{d}{ds}\log(\frac{1+s\sigma}{1+s})     )\|    \sum_{i,j} \| P_{i}X_{j}G\|. $$
The third difference is then bounded by a fixed constant multiple of  $\lambda^{2}\| |\vec{P}|^{n-2}\rho |\vec{P}|^{n-2}\|_{1}\sum_{j} \| |\vec{P}| X_{j}G\|$.
\end{proof}

\appendix

\begin{center}
    {\bf \Large APPENDIX \normalsize }
  \end{center}

\section{Hilbert spaces and operator inequalities}

\begin{lemma}\label{DensityMat}
Let $\eta$ be a trace class operator on $\R^{n}$ with continuous integral kernel $\eta(\vec{k}_{1},\vec{k}_{2})$, $A,A^{\prime}\in GL_{n}(\R)$, and $\vec{a},\vec{a}^{\prime}\in \R^{n}$, then   
$$
\int_{\R^{n}}d\vec{x}\, |\eta(A\vec{x}+\vec{a},A^{\prime}\vec{x}+\vec{a}^{\prime})|\leq \frac{1}{2}(\frac{1}{|\det(A)|}+\frac{1}{|\det(A^{\prime})|})\|\eta\|_{1}.
$$
\end{lemma}
\begin{proof}
Let $\rho=\sum_{j} \lambda_{j}|f_{j}\rangle\langle g_{j}|$, then we have
\begin{multline}\label{Formal}
\int_{\R^{n}}d\vec{x}\,|\rho(A\vec{x}+\vec{a},A^{\prime}\vec{x}+\vec{a}^{\prime})|= \int_{\R^{n}}d\vec{x}\,|\sum_{j}\lambda_{j} f_{j}(A\vec{x}+\vec{a})\,\bar{g}_{j}(A^{\prime}\vec{x}+\vec{a}^{\prime})|\\ \leq \frac{1}{2}(\frac{1}{|\det(A)|}+\frac{1}{|\det(A^{\prime})|}) \|\rho\|_{1},
\end{multline}
where the inequality follows from $2ab\leq a^{2}+b^{2}$ for $a,b\in \R^{+}$ and completing the integration.    The equality on the left-hand side of~(\ref{Formal}) is formal for a general integral kernel $\rho(\vec{k}_{1},\vec{k}_{2})$ which is defined only a.e. with respect to joint integration over $\vec{k}_{1},\vec{k}_{2}$, but with our continuity condition it is well defined.

\end{proof}

\begin{proposition}\label{StrongConv}
For $n\in \N$, let $A_{n}\in \Bi(\Hi)$  for a Hilbert space $\Hi$ and 
$$\frac{1}{2}\sum_{n} |A_{n}|+|A_{n}^{*}|$$
be weakly convergent to a bounded operator with norm $c$.   Then $\sum_{n}A_{n}$ is strongly convergent to a bounded operator $X$ with $\|X\|\leq c$.  
\end{proposition}
\begin{proof}
Let $g\in \Hi$ and with the polar decomposition~\cite{Reed} $A_{n}=U_{n}|A_{n}|$, then taking a tail sum  
\begin{multline*}
\|\sum_{N}^{\infty}A_{n}g\|_{2}= \sup_{\|h\|_{2}=1} \langle h| \sum_{N}^{\infty}U_{n}|A_{n}|^{\frac{1}{2}}|A_{n}|^{\frac{1}{2}} g\rangle \\ \leq    \sup_{\|h\|_{2}=1}   \sum_{N}^{\infty}\||A_{n}|^{\frac{1}{2}}U_{n}h\|_{2}\| |A_{n}|^{\frac{1}{2}}g\|_{2}\leq \sup_{\|h\|_{2}=1}\big(\sum_{n=1}^{\infty}\langle h| |A_{n}^{*}| h\rangle\big)^{\frac{1}{2}}\big(\sum_{N}^{\infty}\langle g||A_{n}| g\rangle \big)^{\frac{1}{2}},  
\end{multline*}
where the first and second inqualities follow by two different applications of the Cauchy-Schwartz inequality.  The right-hand side then tends to zero for large $N$ by our assumptions on the series $\sum |A_{n}|$ and $\sum_{n}|A_{n}^{*}|$.  The operator norm bound can be seen from the same calculation with a sum over all $n$ rather than a tail.

\end{proof}

\begin{proposition} \label{StrongConv2}
Let $A_{n}$, $B_{n}$ for $n\in \N$ be elements in $\Bi(\Hi)$  for a Hilbert space $\Hi$ such that 
$\sum_{n}A_{n}^{*}A_{n}$ and $\sum_{n}B_{n}^{*}B_{n}$ converge weakly to bounded operators with norms less than $c$, then the sum
$$\varphi(G)=\sum_{n}A_{n}^{*}\,G\, B_{n}$$
is strongly convergent to an operator with norm less than or equal to $c\|G\|$. 
\end{proposition}
The proof follows a similar pattern to that of Proposition~\ref{StrongConv}.

\section{Norm bounds for $V_{1}$, $V_{2}$, $\vec{A}$ and $\varphi$}
 The following lemma shows that the limiting expressions~(\ref{Lindblad}) vary continuously with respect to the density operator $\rho $ in the $\|\cdot\|_{wtn}$ topology.   It allows the limiting expression to be defined for all $\rho$ with $\|\rho\|_{wtn}<\infty$ without the additional assumption that the integral kernel of $\rho$ in the momentum representation in continuously differentiable.   A somewhat weaker norm than $\|\cdot\|_{wtn}$ would suffice.

\begin{lemma}\label{LimitExp}
Let $V_{1}$, $V_{2}$, $\vec{A}$, and $\varphi$ be defined as in~(\ref{one})-(\ref{four}) for a $\rho\in \mathcal{B}_{1}(\R^{n})$, $\rho\geq 0$,  with continuously differentiable integral kernel in the momentum representation, then there is a constant $c>0$ such that for all $\rho$ and $j$   
$$\|V_{1}\|, \|V_{2}\|, \|\vec{A}\|, \|[P_{j},A_{j}]\|,  \|\varphi \|\leq c\|\rho\|_{wtn}.$$

\end{lemma}

\begin{proof}
 By an argument similar to~(\ref{DensityMat})
\begin{eqnarray*}
\hspace{.5cm} \|V_{1}\|\leq  \mathit{c}_{n}\| |\vec{P}|\rho\|_{1},\hspace{1.2cm}\| \vec{A}\| \leq  \mathit{c}_{n}\sum_{j}\|  |\vec{P}|^{n-2} X_{j}\rho\|_{1}, \hspace{1.2cm} \|\varphi\|\leq \mathit{c}_{n}^{2}\| |\vec{P}|^{n-2}\rho |\vec{P}|^{n-2}\|_{1},\\ \|V_{2}\|\leq \mathit{c}_{n}\sum_{j}\| \{ P_{j}|\vec{P}|^{n-1},[X_{j},\rho]\} \|_{1},\hspace{1cm} \| [P_{j},A_{j}]\|\leq c_{n}\| [P_{j}|\vec{P}|^{n-2},[X_{j},\rho]]\|_{1},\hspace{.5cm}
\end{eqnarray*}
where we have used that $\|\varphi\|=\|\varphi(I)\|$ since $\varphi$ is a positive map.   By $\rho$ being self-adjoint, we have inequalities such as  
$$\| P_{j}|\vec{P}|^{n-2}\rho X_{j}\|_{1}\leq  \frac{1}{2}(\|X_{j} \rho X_{j}\|_{1}+\| P_{j}|\vec{P}|^{n-2} \rho  |\vec{P}|^{n-2}P_{j}\|_{1}),$$
and then that $P_{j}\leq |\vec{P}|$.  Finally, since $\rho$ is positive $\int_{\R^{n}}d\vec{k}\rho(\vec{k},\vec{k})|\vec{k}|^{2s}=\| |\vec{P}|^{s} \rho |\vec{P}|^{s}\|_{1}$, and inequalities of the form 
$$\| |\vec{P}|^{r} \rho |\vec{P}|^{r}\|_{1}\leq  \|\rho\|_{1}+\| |\vec{P}|^{s} \rho |\vec{P}|^{s}\|_{1}$$
follow, where $r,s$ have the same sign and $|r|\leq |s|$.   Hence $\|\rho\|_{wtn}$ bounds the expressions for $V_{1}$, $V_{2}$, $\vec{A}$, $[P_{j},A_{j}]$ and $\varphi$.

\end{proof}

\section*{Acknowledgments}
   I thank  Bruno Nachtergaele for discussions on this manuscript.   This work is partially funded  from the Belgian Interuniversity Attraction Pole P6/02.  Also financial support has come from Graduate Student Research (GSR) fellowships  funded by the National Science Foundation (NSF  \# DMS-0303316 and DMS-0605342).

\end{document}